\documentclass[10pt,conference,letterpaper]{IEEEtran}
\IEEEoverridecommandlockouts
\usepackage{array}
\usepackage{stfloats}
\usepackage{url}
\usepackage{amsmath,amssymb,amsfonts}
\usepackage{mathrsfs}
\DeclareMathOperator*{\argmax}{arg\,max}
\usepackage{algorithmic}
\usepackage{algorithm}
\usepackage{graphicx}
\usepackage{textcomp}
\usepackage{xcolor}
\usepackage{multirow}
\usepackage{booktabs} 
\usepackage{color}
\usepackage{color,soul}
\usepackage{hhline}
\usepackage[tight,TABTOPCAP]{subfigure}
\usepackage{caption}
\usepackage{autobreak}
\usepackage{rotating}
\usepackage{enumitem}
\usepackage{amsthm}
\usepackage{verbatim}
\allowdisplaybreaks
\usepackage{pdfpages}
\usepackage{epsfig}
\usepackage{epstopdf}
\usepackage{rotating}
\usepackage{bm}
\usepackage{balance}
\def\BibTeX{{\rm B\kern-.05em{\sc i\kern-.025em b}\kern-.08em
    T\kern-.1667em\lower.7ex\hbox{E}\kern-.125emX}}

\newtheorem{theorem}{Theorem}

\newtheorem{lemma}{Lemma}

\newtheorem{Definition}{Definition}

\def\mathbi#1{\textbf{\em #1}}

\newcommand\qiu[1]{{\color{black} #1}}
\newcommand\eat[1]{}

\begin{document}
\title{\huge Bandit-Based Charging with Beamforming for Mobile Wireless-Powered IoT Systems}

\author{Chenchen Fu, Zining Zhou, Xiaoxing Qiu, Sujunjie Sun, Weiwei Wu, and Song Han\thanks{C. Fu, Z. Zhou, X. Qiu, S. Sun and W. Wu are with the Department of Computer Science and Engineering, Southeast University, Nanjing, China (E-mail: \{chenchen\_fu, 220212008, xxqiu, sjjsun, weiweiwu\}@seu.edu.cn).}\\
\thanks{S. Han is with the School of Computing, University of Connecticut, United States (E-mail: song.han@uconn.edu).}
}

\maketitle

\begin{abstract}
Wireless power transfer (WPT) is increasingly used to sustain Internet-of-Things (IoT) systems by wirelessly charging embedded devices. Mobile chargers further enhance scalability in wireless-powered IoT (WP-IoT) networks, but pose new challenges due to dynamic channel conditions and limited energy budgets. Most existing works overlook such dynamics or ignore real-time constraints on charging schedules.
This paper presents a bandit-based charging framework for WP-IoT systems using mobile chargers with practical beamforming capabilities and real-time charging constraints. We explicitly consider time-varying channel state information (CSI) and impose a strict charging deadline in each round, which reflects the hard real-time constraint from the charger's limited battery capacity.
We formulate a temporal-spatial charging policy that jointly determines the charging locations, durations, and beamforming configurations. Area discretization enables polynomial-time enumeration with constant approximation bounds. We then propose two online bandit algorithms for both stationary and non-stationary unknown channel state scenarios with bounded regrets.
Our extensive experimental results validate that the proposed algorithms can rapidly approach the theoretical upper bound while effectively tracking the dynamic channel states for adaptive adjustment.
\end{abstract}

\section{Introduction}
{
IoT systems powered by energy harvesting techniques have drawn growing attention from both academia and industry in recent years. Such systems, particularly those built on embedded platforms, are designed to operate autonomously with minimal human intervention. However, IoT systems face unique constraints, including limited processing power, strict energy budgets, real-time communication requirements, and physical size limitations. These challenges make efficient power management and optimized communication strategies essential for the sustainability and scalability of IoT applications.  
Compared to harvesting energy from ambient sources~\cite{Huang2023Adaptive}, such as solar, wind, and vibration, which are typically unpredictable and highly variable, a more controllable and stable approach is Radio Frequency (RF)-based Wireless Power Transfer (WPT)~\cite{Zeng2017Communications}. WPT enables the IoT devices to wirelessly harvest energy from dedicated RF sources, mitigating or even eliminating the need for frequent battery replacements. However, a fundamental limitation of WPT is that RF signal strength decays rapidly with distance, leading to inefficient energy transfer, which is particularly problematic for low-power IoT devices operating in distributed environments~\cite{choi2016wireless}.
To address this problem, {\em beamforming techniques} play a critical role in enhancing the energy transfer efficiency by concentrating the RF energy on specific direction(s), thus minimizing the energy loss and extending the distance of the energy transfer~\cite{Zeng2017Communications}. 

}

{
A series of work in the literature focused on beamforming technique design and waveform optimization~\cite{huang2017large} for WP-IoT systems with static charger(s). 
Practical charging algorithms based on beamforming techniques were designed and implemented in~\cite{choi2016wireless,choi2018theory}.
However, due to the limited range of wireless power transfer, static chargers have to be deployed in proximity to the target devices. 
In large-scale IoT systems, it is however economically infeasible to deploy a large number of static chargers in the field to cover all the devices.
An effective way to address this problem is to deploy {\em mobile chargers} to travel through the field and approach individual devices to provide wireless energy transfer. In recent years, researchers focused on trajectory planning and charging algorithm designs to charge the devices using mobile chargers \cite{liu2020effective}. 
{However, in most of these studies, {\em their charging models did not consider the beamforming technique (or over-simplified it), neglected the unstable and dynamic channel conditions, and failed to account for real-time constraints imposed by limited energy budgets}, and thus cannot be directly deployed in real-world scenarios.}
Their common limitations include but are not limited to: {\bf 1)} not considering the multi-antenna technology to enhance the energy transfer efficiency \cite{dai2018wireless,dai2021placing,wang2019robust,xu2018maximizing}; {\bf 2)} assuming that the charging power of different beams can be directly additive~\cite{dai2018wireless}; {\bf 3)} simplifying the charging region from the irregular droplet shape (see Fig.~\ref{fig:eg_beam} and Fig. \ref{fig:multi}) to the regular circular shape~\cite{dai2018wireless,dai2021placing};
and {\bf 4)} not taking into account the dynamic channel condition (see Section \ref{sec:pre} for details), which plays a significant role in most wireless charging problems with either static or mobile chargers \cite{dai2018wireless,wang2019robust,dai2021placing,jiang2017joint,xu2018maximizing,liu2020effective}.
}

In this work, we study the fundamental problem of effective wireless charging with mobile charger(s) in WP-IoT systems, by utilizing practical beamforming techniques under dynamic channel condition. {\textit{The objective is to maximize the charging utility for all target devices to avoid starvation and energy overflow problems (see Def. \ref{def:u}), while satisfying the real-time charging requirement introduced by the limited battery of the mobile charger(s).}}
As wireless charging suffers severe swings of channel states and low charging efficiency, the mobile charger(s) should stop at carefully determined locations to charge the devices in proximity. At each location, a proper beamforming code word (See Def. \ref{def:codeword}) should be selected to maximize the total charging utility.
To achieve this objective, we need to tackle the following key challenges.

\begin{enumerate}[leftmargin=*, labelindent=0pt]
    \item {\bf [Charging objective]} As both starvation and energy overflow problems need to be avoided for all devices, the charging objective should be well designed to achieve submodularity
    \footnote{\qiu{A set function is submodular if it has diminishing returns: adding an item to a smaller set gives more gain than to a larger one. In our case, charging a low-energy sensor brings more utility than charging a well-powered one.}}
    ({\em charging starving devices is more crucial than charging devices with adequate energy}). The objective should also be generally applicable to a variety of real-world systems with unique requirements. {Moreover, the limited battery capacity of mobile charger(s) leads to tight time constraints for each charging round. To serve multiple energy-demanding sensor nodes under such constraints, it is crucial to develop effective charging objectives and scheduling strategies that can maximize overall efficiency.}
    \item {\bf [Where to charge]} Determining the candidates of charging locations involves segmenting the area into grids.
    By modelling the power distribution in practical charging beams (See Fig.~\ref{fig:multi}), we observe that the power intensity varies significantly even within a single grid. {\em This variation is not uniform or linear as assumed in previous work, but needs to be well analyzed to obtain the approximation ratio introduced by area discretization}, which in turn affects the overall efficiency and accuracy of the charging strategy.
    \item {\bf [How to charge]}
    The target charging problem involves submodular objective function. An oracle to adeptly handle this submodular charging problem, and {\em the bandit schemes with tight regret bound are not analyzed in the literature and remains an open problem}. Novel bandit schemes and the regret bounds should be well designed and analyzed based on the new oracle, ensuring that the approach is not only theoretically grounded but also pragmatically viable.
    
\end{enumerate}

\noindent To the best of our knowledge, this work is the first attempt to explore the wireless charging problem in WP-IoT systems with mobile charger(s) based on practical beamforming model. The contributions are summarized below. 
\begin{itemize}[leftmargin=10pt]
\item We analyze the key features of beamforming and the key requirements of device charging in WP-IoT systems to define the general formulation of charging utility function. 
\item We conduct area discretization based on practical beam models, so that the charging locations can be enumerated in polynomial-time with constant approximation ratio. 


\item \qiu{By assuming CSI is known, a scheme called {\bf Greedy Utility Algorithm (GUA)} which achieves $(\frac{1}{2}-\frac{1}{2e})-$ approximation is first proposed as oracle. When CSI is unknown, two novel bandit-based algorithms, {\bf Utility Maximization Combinatorial Bandit (UMCB)} and {\bf Sliding-Window UMCB (UMCB-SW)}, are further proposed for the cases where the unknown CSI is fixed and dynamic, respectively. The proposed bandit schemes utilize GUA to select arms based on the estimated energy status. The regret of UMCB is proved to achieve the tight regret bound $O(ln T)$, where $T$ represents the number of charging rounds.}
    
\end{itemize}

\section{Preliminaries, Models and Problem Formulation}\label{sec:model}

In this section, we first describe the preliminaries of this work.  We then introduce the practical system models employed in this work, and finally formulate the target problem. 

{
\subsection{Preliminaries}
\label{sec:pre}
\vspace{0.05in}
\noindent {\bf[Beamforming with multiple antennas]}: 
Beamforming is a signal processing technique that directs the transmitted signal toward a specific receiver while minimizing interference to other users.
In wireless energy transfer, beamforming plays a crucial role in directing the energy beam to the receiver, and significantly improving the charging performance.
A comparison of the beamforming array and conventional array is presented in Fig.~\ref{fig:eg_beam}.
Rather than simply broadcasting energy or signals in all directions, beamforming utilizes an array of antennas that can be steered to transmit radio signals in a specific direction.
The antenna array determines the direction of interest and transmits a stronger beam of signals in that direction, forming a distinctive droplet-shaped power region due to the constructive and destructive interference of multiple omnidirectional antennas. By leveraging multiple antennas at the transmitter, beamforming enhances signal strength at the intended destination, improving {energy transfer efficiency} and reliability.

\vspace{0.05in}
\noindent {\bf[Dynamic channel condition]}:
To implement effective beamforming, the energy transmitter needs knowledge of the wireless channel conditions, which is referred to as channel state information (CSI). 
CSI provides insights into how signals propagate through the environment, allowing the transmitter to adjust its beamforming strategy accordingly. 
{Accurate CSI is essential for optimizing the beamforming strategy to ensure effective energy reception. However, obtaining it is challenging in practical systems due to the instability of the wireless channel, especially in fast-changing environments.}
Direct feedback of full CSI from the receiver to the transmitter can result in excessive overhead. Therefore, modern systems often utilizes the codebook to reduce this burden.

\begin{figure}[t]
\centering 
    \includegraphics[width=0.8\linewidth]{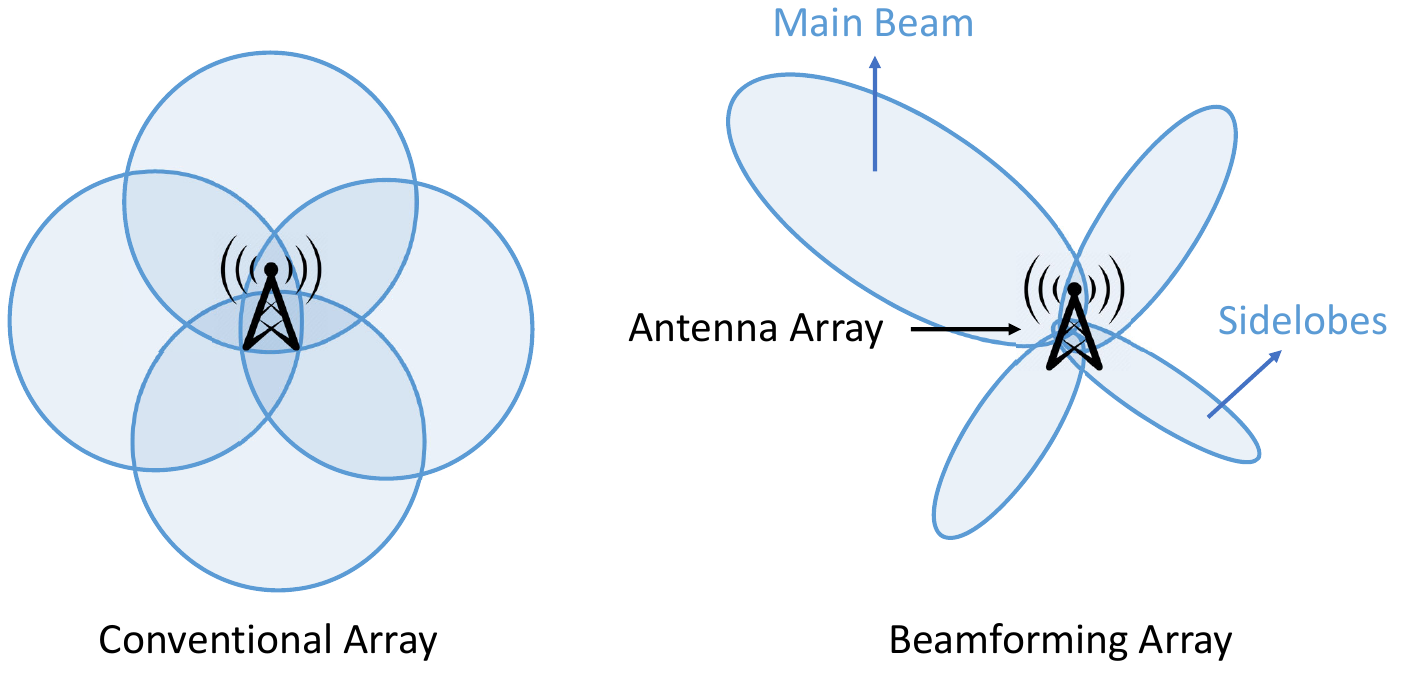} 
    \caption{Comparison of Conventional and Beamforming Array} 
    \label{fig:eg_beam}
\end{figure}

\vspace{0.05in}
\noindent {\bf[Codebook and code word]}: A codebook is a predefined set of beamforming vectors or matrices that the transmitter uses to select the optimal beamforming direction based on limited feedback from the receiver. Instead of transmitting full CSI, the receiver sends back only an index corresponding to the best-matching entry in the codebook, thus significantly reducing the required feedback. Each entry in the codebook is called a code word, representing a specific beamforming weight or precoding vector.

By leveraging CSI-aware beamforming with well-designed codebooks, wireless networks can quantize channel state information into a finite set of predefined beamforming vectors, thereby enabling limited feedback-based beam selection, minimizing interference, and improving charging efficiency. However, such codebook-based limited feedback beamforming has not been fully explored in wireless-powered IoT (WP-IoT) systems. Motivated by this gap, this paper investigates an efficient beamforming-based charging approach for battery-constrained IoT devices, focusing on optimizing wireless charging with mobile chargers under dynamic and unknown channel conditions to enhance energy transfer performance.



}

\begin{figure}[t]
\centering 
    \includegraphics[page=1,width=0.99\linewidth]{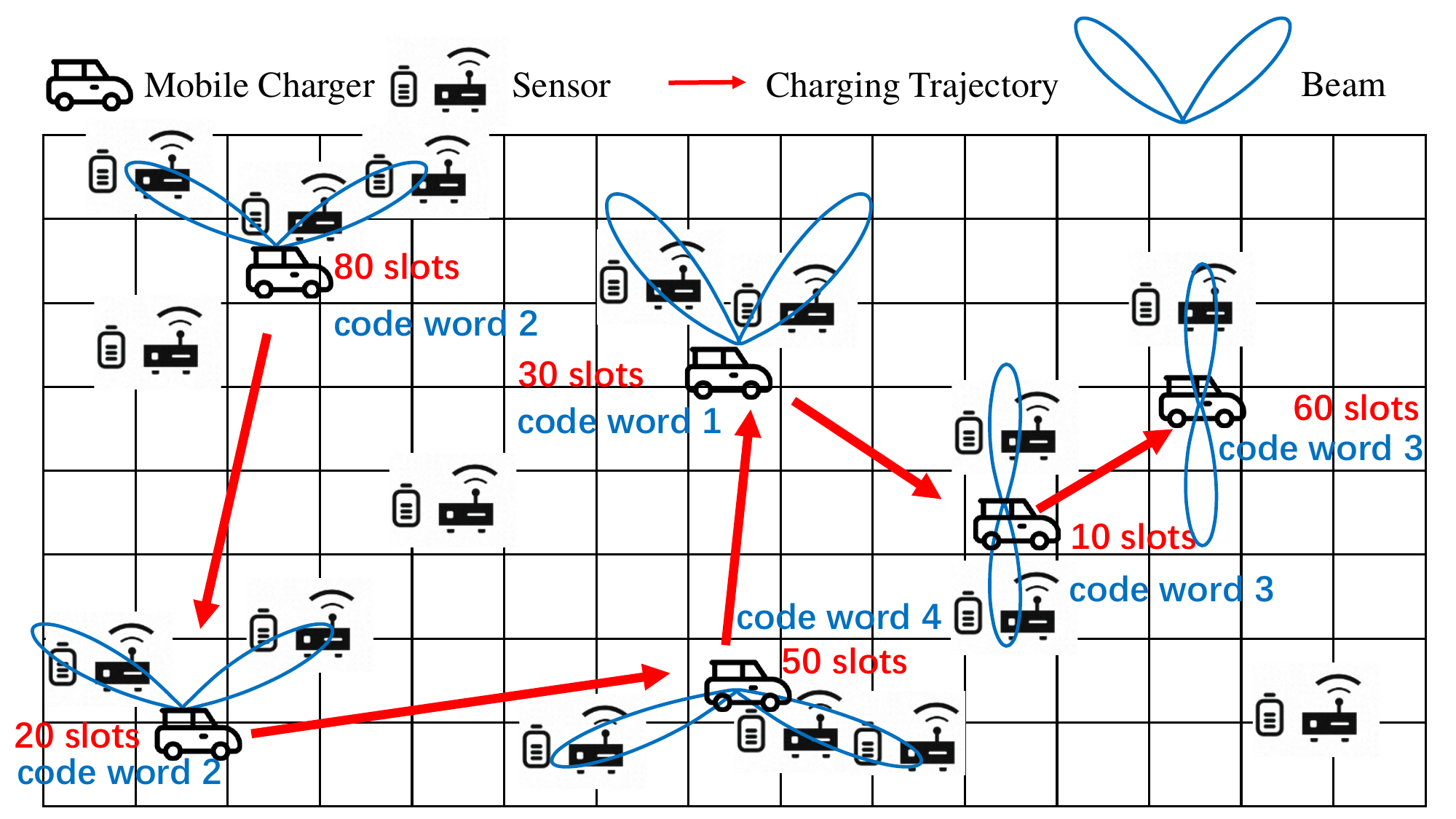} 
    \caption{System model for a WP-IoT system with a mobile charger. (X slots, code word Y) means the mobile charger stops at the selected location and uses code word Y to charge for X time slots.} 
    \label{fig:system model}
\end{figure}

\subsection{System Model}
\label{sec:system_model}
In this work, we consider a WP-IoT system that comprises $N$ fixed sensors and a mobile charger. The proposed algorithms are also applicable to multiple mobile chargers as to be discussed in Section~\ref{sec:discussion}. An overview of the system architecture is shown in Fig.~\ref{fig:system model}, where the mobile charger(s) traverses through the sensors to provide wireless energy transfer.
The mobile charger is equipped with a uniform linear antenna array (ULA) of $N_{a}$ identical antennas and the corresponding RF chains.
The RF chain allows the user to regulate the phases of the channel and then amplify the signal, thereby enables the beamforming technique~\cite{yun2021towards}.

\vspace{0.05in}
\noindent {\bf [Beamforming Model]:} We consider a practical beamforming model~\cite{Balanis2005Antenna, choi2016wireless} in our study. The simulation results of normalized power with different beamforming vectors is shown in Fig.~\ref{fig:multi}.
With a fixed charging distance, we obtain the normalized power by setting the beam directing to $45^{\circ}$ in the polar coordinate system.  The power region is of the droplet shape, showing the directional wave formed by multiple omni-directional antennas. And the power intensity in the (two) main beam directions dominate those in other directions. 
{In this work, we analyze the power gain in the main beam direction given a code word. The formal definition of code word is provided in Def.~\ref{def:codeword}.}
\emph{Different from most existing studies (either with static chargers~\cite{dai2018wireless,wang2019robust,dai2021placing} or mobile chargers~\cite{jiang2017joint,xu2018maximizing,liu2020effective}) that assume uniform or linear power distribution in the charging region of regular circular shape, this work considers practical beamforming charging model for the development of real systems.}

\begin{figure}[t]
    \subfigure[code word directing to $45^{\circ}$ ]{
    \includegraphics[width=0.46\linewidth]{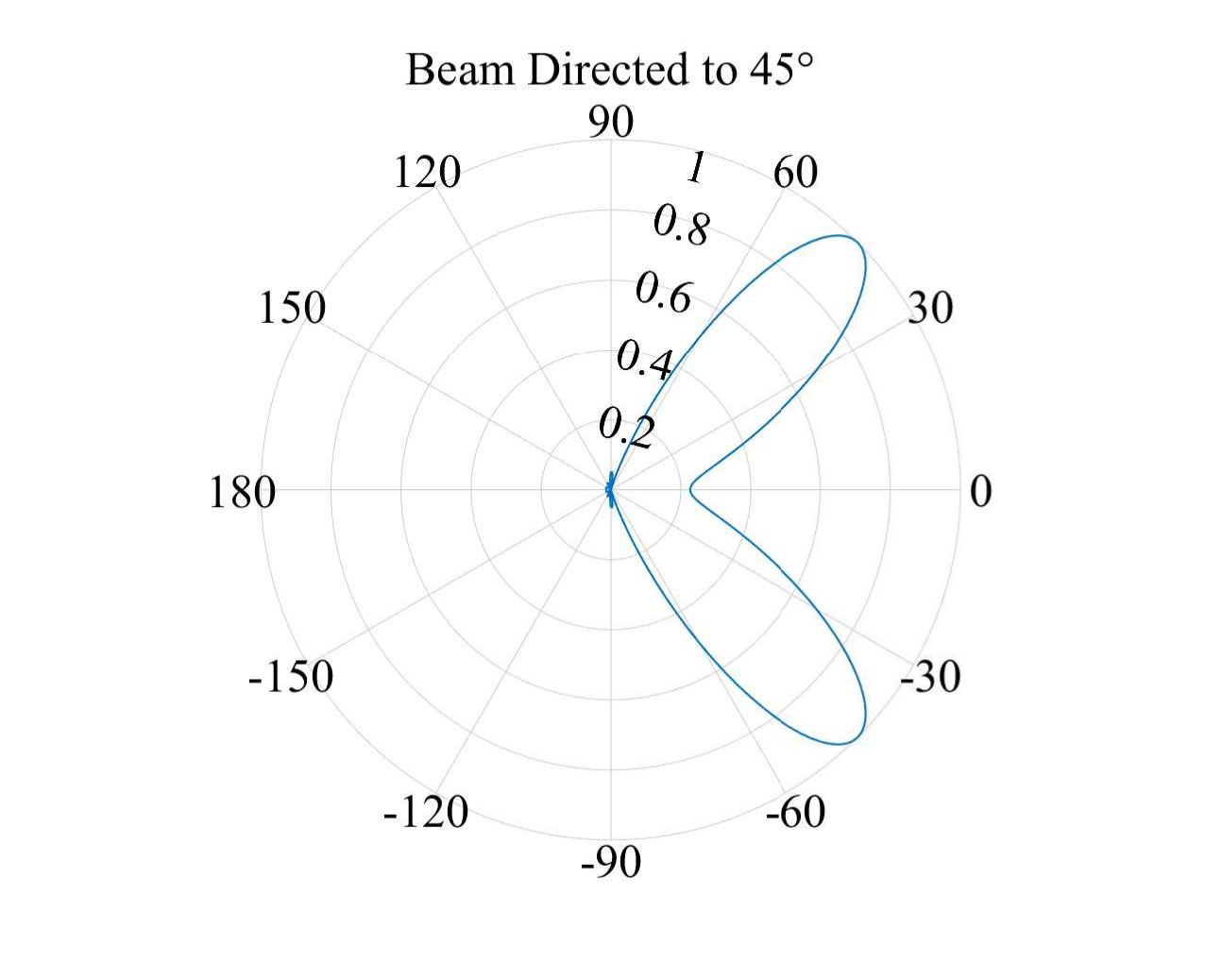}
    \label{fig:45_polar}
    }
    \subfigure[code word directing to $45^{\circ}$ ]{

    \includegraphics[width=0.46\linewidth]{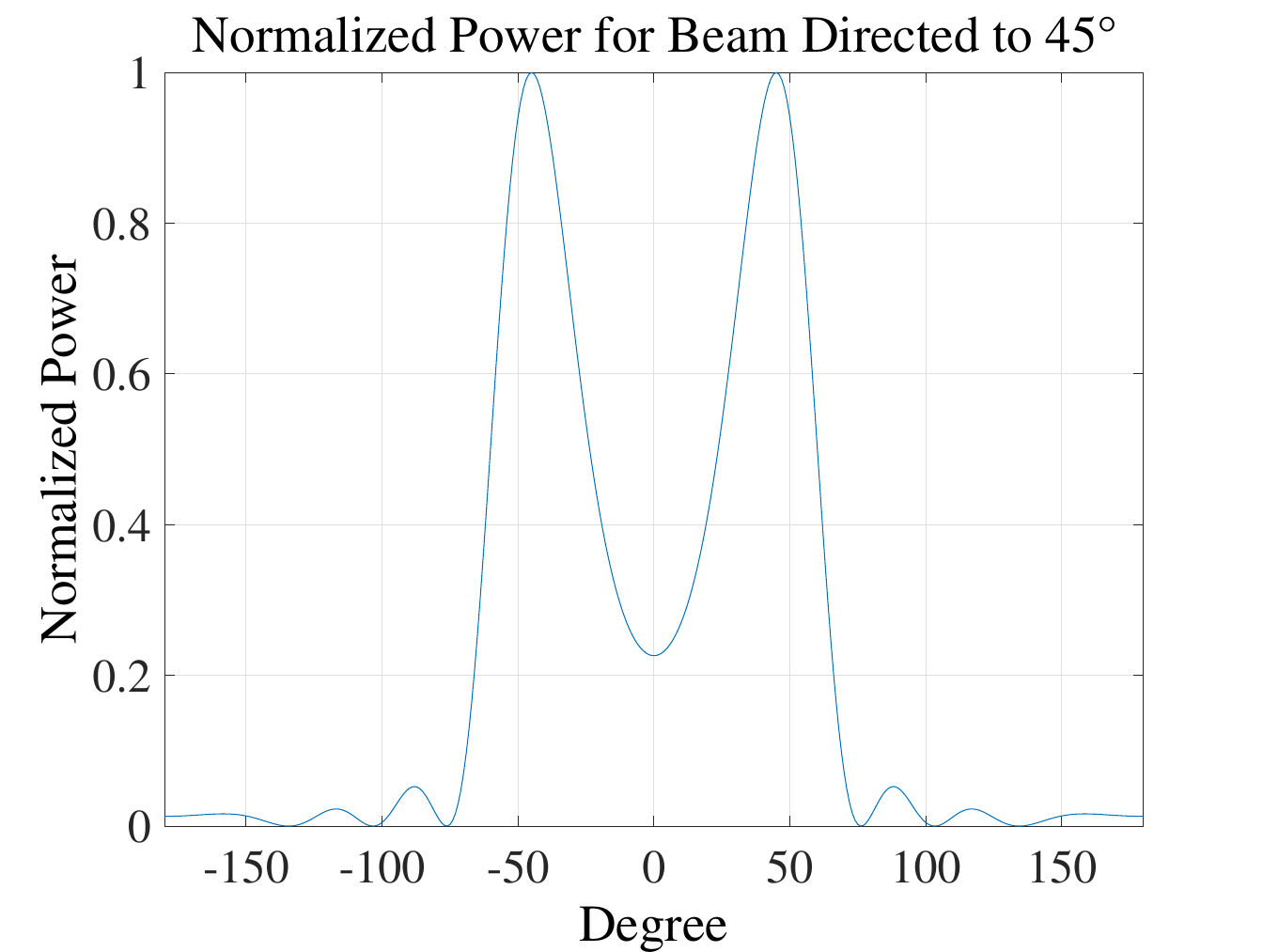}
    \label{fig:45_square}
    }
    
    \caption{
    Normalized power in azimuth cut (obtained by setting the 800MHz wave from 8-antenna ULA with a spacing of 0.1m).}
    \label{fig:multi}
\end{figure}

\begin{figure}[t]
    \subfigure[Received power from a fixed distance of 1.5m]{
    \includegraphics[width=0.45\linewidth]{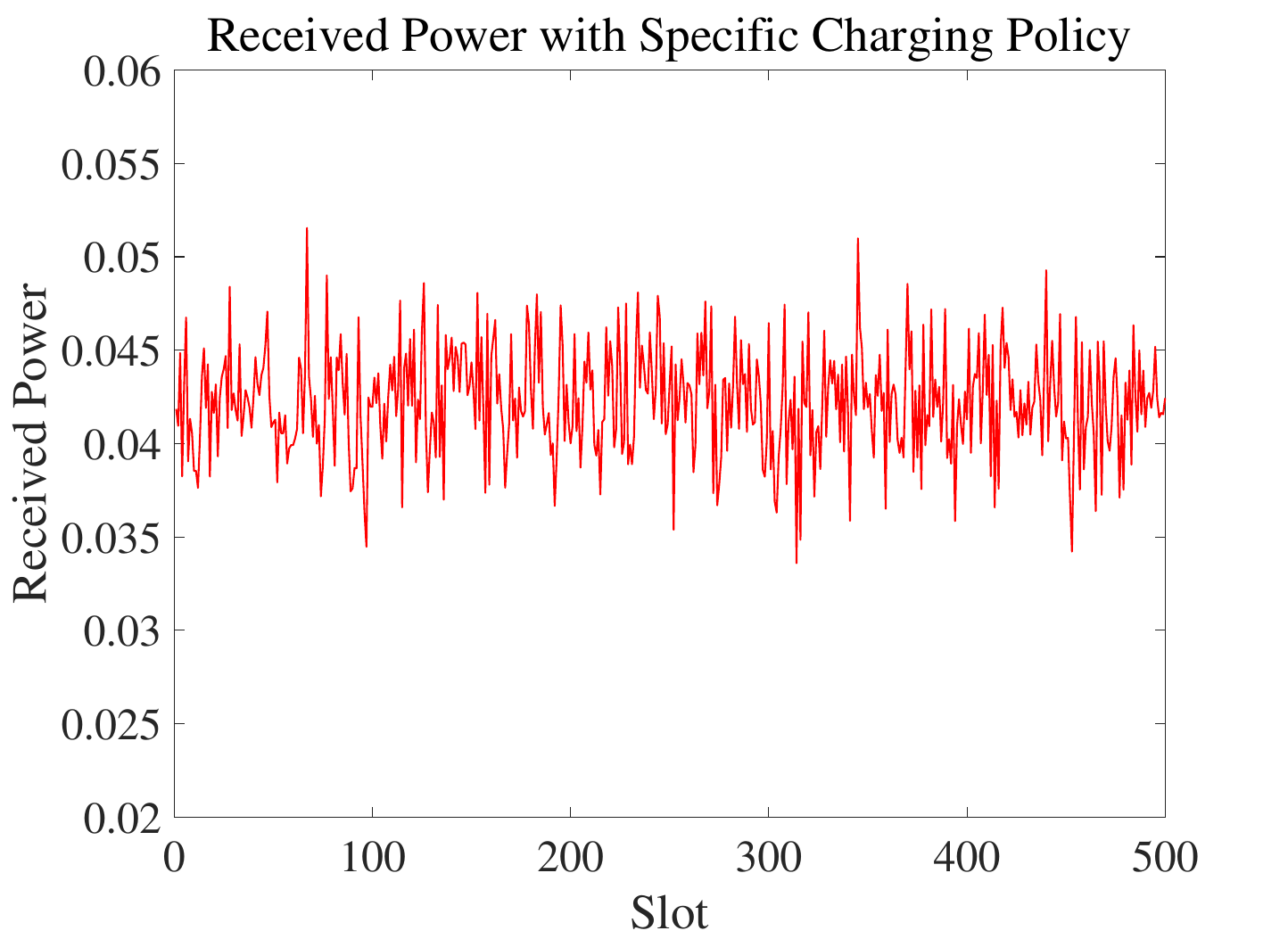}

    \label{fig:dynamic_channel}
    }
    \subfigure[Received power from varied distances]{
    \includegraphics[width=0.45\linewidth]{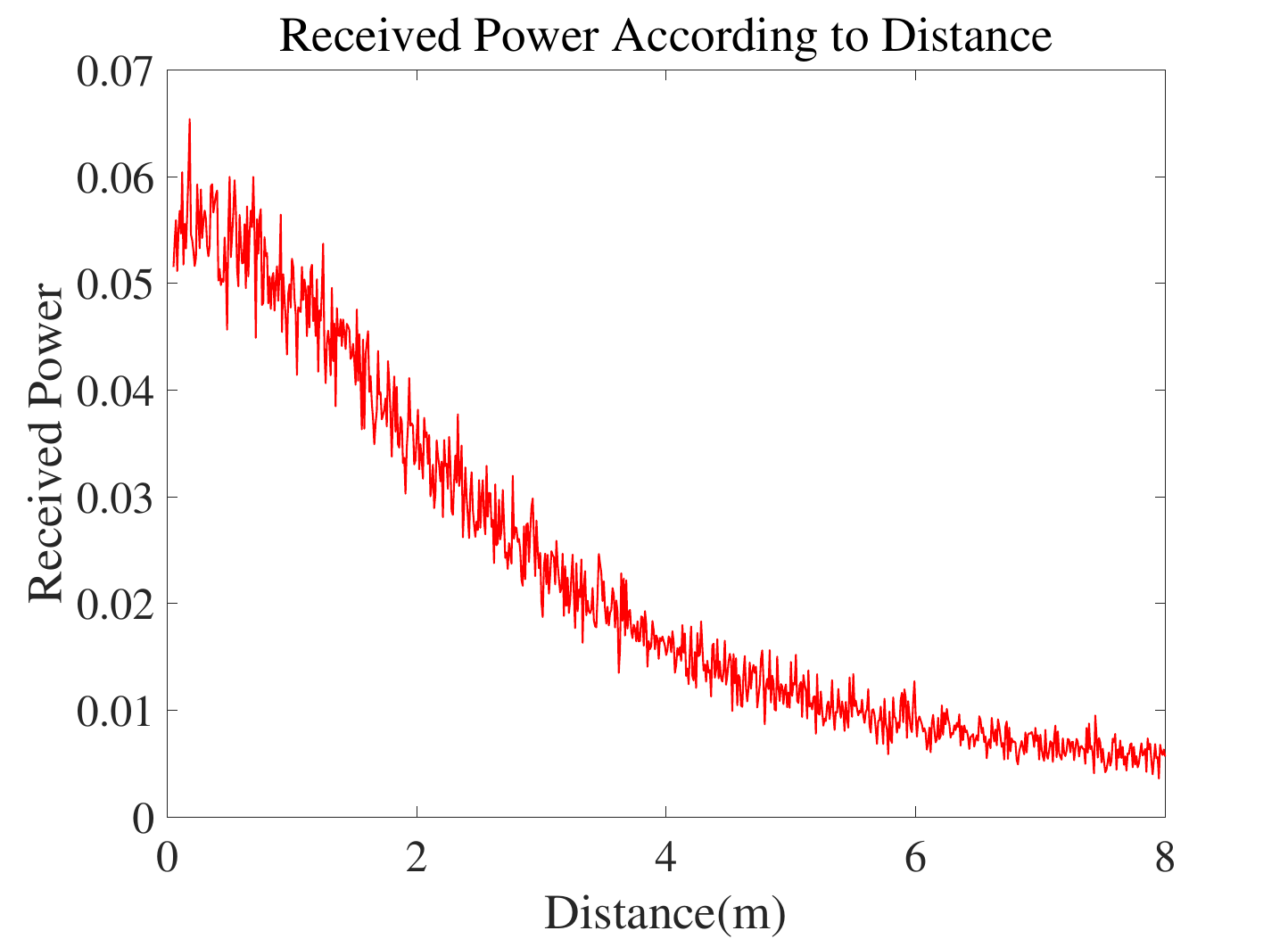}
    \label{fig:dynamic_dis}
    }
  \centering  
  
    \caption{
    Received power from static and mobile chargers simulated with the perfect beamforming vectors through the 800MHz wave from 8-antenna ULA with a spacing of 0.1m.}
\end{figure}

\vspace{0.05in}
\noindent {\bf [Dynamic and Unknown Channel State Information (CSI)]:}
In the path-loss channel model like two-ball model~\cite{di2015stochastic} or the classic Rayleigh model~\cite{yu2023reassessment}, the intensity of the received signal is proportional to the channel gain, which is characterized as an exponentially distributed random variable possessing a mean value of unity, indicating the dynamic fluctuation of channel state.
As shown in Fig.~\ref{fig:dynamic_channel}, the received power fluctuates from slot to slot due to the dynamic CSI, given fixed charging distance.
When the charging distance increases, as shown in Fig.~\ref{fig:dynamic_dis}, the received power also fluctuates significantly.
In fact, in real-world systems, besides the dynamic nature of CSI, we even do not know the mean value of channel gain before one actual charging begins. 
\emph{While most existing studies \cite{dai2018wireless,wang2019robust,dai2021placing} proposed solutions assuming that CSI is known and stable, this work explores practical charging policy with dynamic and unknown CSI, to be more applicable in real-world scenarios.}


\vspace{0.05in}
\noindent {\bf [Sensors]:} Let $S=\{s_1,\dots,s_{{N}}\}$ be the set of sensors to charge. All sensors are deployed in a 2D area with given coordinates.
Each sensor has one antenna for harvesting RF energy and another antenna for exchanging information with the mobile charger (e.g., energy status). These two antennas operate on different frequency bands and have no interference.
Each sensor is equipped with a battery with a capacity of $Q$\footnote{We assume that all sensors have the same battery capacity, which can be easily extended to support different battery capacities for individual sensors.}. 
We use $x_i$ to denote the {initial energy} of sensor $s_i$ before a charging round and we have $X=$ $\{x_1,...x_i,...x_N\}$. $p_i^c(t)$ is the power consumption function of sensor $s_i$. \emph{To satisfy the practical requirement, $p_i^c(t)$ is unknown and can be set according to real system features.}

{
\vspace{0.05in}
\noindent {\bf [Charger]:}
The mobile charger traverses through the area and charges the sensors in rounds to keep the WP-IoT system operate properly. As CSI changes rapidly during the movement of the charger~\cite{KimClerckx2020Signal}, the charging efficiency can be significantly reduced. As a result, in this work charging is only allowed when the mobile charger stops at some locations~\cite{liu2020effective}. We assume that the total moving time of the mobile charger is negligible compared to the stop-to-charging time spent on the sensors, which can be easily demonstrated in practical systems, as the wireless charging efficiency is low due to strict regulatory limitations. The orientations of the charger are implicitly handled by the beamforming codeword design, thus its rotational posture need not be explicitly modeled.

\vspace{0.05in}
\noindent {\bf [Real-Time Charging Constraint]:}
Each charging round is subject to a strict temporal constraint, denoted as $T_c$, which represents the maximum allowable duration for completing all charging activities in one round. This constraint arises from the limited battery capacity of the mobile charger, which restricts the total energy budget and available charging time. Importantly, \textit{$T_c$ serves as a hard real-time deadline}, within which all selected sensors must be charged based on their spatial distribution, residual energy, and communication demands.
Failure to meet this deadline may lead to sensor energy depletion and service degradation, making $T_c$ the key constraint that governs the design of the charging schedule. Therefore, \textit{charging locations, stop durations, and beamforming configurations must be jointly optimized to ensure the feasibility of charging tasks under the hard deadline $T_c$.}
}

\vspace{0.05in}
\noindent {\bf {[Charging Policy Design]}:} Based on the above model, in each charging round, we need to design the proper charging policy, which determines {\bf where} (the charging locations for the mobile charger), {\bf when} (the charging intervals in each location) and {\bf how} (the specific beam formed) to charge the sensors.
The moving trajectory of the mobile charger thus can be obtained by connecting the charging locations in the order of the corresponding charging intervals at each location.

The mobile charger maintains a codebook $\mathcal{M}=\{\bold{w}_1,\cdots,\bold{w}_M\}$, where the total number of code words (See Def.~\ref{def:codeword}) is $M$, to determine how to perform beamforming at each charging location.
When the charging locations and the corresponding code words are determined, the signal received at sensor $s_i$ can be represented as follows:
\begin{equation}
\label{eq:1}
    y_i=\bold{h}_i^{T}\bold{w},
\end{equation}
where $\bold{h}_i$ denotes the CSI between the charger and sensor $s_i$. The expectation of CSI is $\mathbb{E}[\bold{h}]_i \sim \frac{1}{d_i}\in\mathbb{R}^{1\times N_{a}}$ \cite{Balanis2005Antenna,choi2018theory}, where $d_i$ is the distance between $s_i$ and the charger. $\bold{w}$ is the code word picked by the mobile charger.
In most real systems, $\bold{h}_i$ is unknown but follows the Gaussian distribution~\cite{zhao2021online}.
Based on Eq.~(\ref{eq:1}), the received power at sensor $s_i$ is:\footnote{In real-world systems, the amount of power harvested by the sensor is typically smaller than the received power due to the power transfer loss from the RF signal to DC power. W.l.o.g. we denote the power transfer efficiency as $\eta,~\eta\in(0,1]$ \cite{bacinoglu2018energy,li2018energy}, and thus the harvested power is $\eta \cdot p_i$. In this work, for simplicity we assume that $\eta=1$ but the proposed algorithm and theoretical results in this work hold for different values of $\eta$.}
\begin{equation}
\label{eq:pi}
p_i(d_i,\bold{w})=|y_i|^2=|\bold{h}_i^T\bold{w}|^2.
\end{equation}



\begin{Definition}{\bf [Code word]}
\label{def:codeword}
A code word is a specific beamforming weight vector $\bold{w}=(w_1,\dots,w_n,$
$\dots,w_{N_{a}})$ \cite{van2002optimum}, where $w_n$ is a complex value indicating the phase and the amplitude of the signal from the $n$-th antenna of the charger.
\end{Definition}

\subsection{Problem Formulation}

Based on the system models, this work aims to solve the {\bf U}tility {\bf M}aximization Wireless {\bf C}harging Problem with mobile charger (UMC problem): 
{\qiu{Consider a charging round consisting of $N_s$ slots and $N$ sensors, where each sensor $s_i$ has an initial energy level $x_i$. For each slot, we select a charging policy $\pi_j = \{l_j, \bold{w}_j\}$ from all policies,}
i.e. the mobile charger stops at a location $l_j$ and picks a code word $\bold{w}_j$ for charging in each slot, where $j\leq |\Gamma|$ and $\Gamma=\{\pi_1,\pi_2,...\}$ is the set of all policies, i.e. $\Gamma=L \times \mathcal{M}$, where $L$ is the set of all available locations for the mobile charger to stop, and $\mathcal{M}$ is the codebook.
The objective of the UMC problem is to maximize the charging utilities of all sensors at the end of each round.

\noindent {\bf [Utility]:} The utility of a sensor should be proportional to its energy stored in battery after charging, denoted as $q_i$, and also needs to be defined to avoid both of the energy overflow given the battery capacity $Q$ and the starvation in the long-term system operation.
And to be applied in general systems, 
in the following we provide the general form of the utility function.
Specific utility functions (see examples in Section \ref{sec:exp_set}) can be designed according to the practical system requirements (individual power consumption, energy characteristics, etc.) as long as the the following definition (Def. \ref{def:u}) is satisfied.

\begin{Definition}{\bf [Utility Function $U(q_i)$]}
\label{def:u}
    $U(q_i)$, where $q_i$ is the energy of sensor $s_i$ after charging, can be any non-decreasing, non-negative, smooth
    \footnote{We utilize the Lipschitz continuity condition by the constant $B$ to describe the smoothness of $U(q_i)$. The condition indicates that there exists $B>0$ such that for any $q_1,q_2\in [0,Q]$, we have $U(q_1)-U(q_2)\leq B\frac{q_2-q_1}{Q}$. For a specific utility function $U(q_i)$, we are able to judge the performance of the scheme by giving $B$ as the lower bound, which will be discussed in online scenario when measuring the performance of an online learning algorithm.}
    and {convex function with a decreasing first derivative}. 
\end{Definition}

Based on the definition, the utility function has submodularity. This guarantees that we can give higher reward to charging the starving sensors and give less or no reward to charging those with adequate energy in battery. The target UMC problem is formulated as follows.

\vspace{0.05in}
\noindent \mathbi{Objective:}
\qiu{
\begin{equation}
\label{fr:convex1}
\max_{\{z_j^n, \pi^n\}} \sum_{i=1}^{{N}}{U(q_i^{N_s})}-\sum_{i=1}^{{N}}{U(C_i^{N_s})}
\end{equation}
}
\noindent $\mathbi{Subject to:}$
\begin{equation}
    \label{eq:realtime}
    N_s\cdot t_u \leq T_c
\end{equation}
\begin{equation}
    \label{eq:consumed}
C_i^n= x_i-\sum_{s=1}^{n}p^c_i(n)t_u,~\forall n\in[N_s]
\end{equation}
\begin{equation}
    \label{eq:calq_i}
q_i^n =C_i^n + \sum_{s=1}^{n} p_i(\pi^s)t_u,~\forall n\in[N_s]
\end{equation}
\begin{equation}
\begin{array}{ll}
    \label{eq:rangeq_i}
0\leq q_i^n \leq Q,~\forall n\in[N_s]
\end{array}
\end{equation}
\begin{equation}
\begin{array}{ll}
    \label{eq:z_j}
\pi^n=\sum_{j=1}^{|\Gamma|} z^n_j \pi_j, ~\forall n\in [N_s],~\forall \pi_j\in \Gamma
\\
\sum_{j=1}^{|\Gamma|} z^n_j=1,~z^n_j\in \{0,1\},~\forall n\in [N_s],~\forall j\in [|\Gamma|]
\end{array}
\end{equation}
In Eq. (\ref{fr:convex1}), $q^n_i$ denotes the energy in the battery of $s_i$ (after charging) at the $n$-th time slot, where $n\in [1,N_s]$, and $N_s$ is total number of slots in a charging round. Thus $q^{N_s}_i$ represents the energy of $s_i$ at the end of a charging round.
$C_i^n$ is the residual energy of sensor $s_i$ at the $n$-th slot if $s_i$ is not charged yet. So $C_i^{N_s}$ is the residual energy of $s_i$ if it is not charged during the entire charging round.
In Eq. (\ref{eq:realtime}), $t_u$ is the duration of one time slot, and $T_c$ is the limited charging time constrained by the limited charging battery of the mobile charger.
In Eq. (\ref{eq:consumed}), 
$x_i$ is the initial energy of $s_i$.
$p_i^c(n)$ is the energy consumed by $s_i$ in the $n$-th slot.
In Eq.~(\ref{eq:calq_i}), $\pi^s$ denotes the policy adopted in the $s$-th slot and the received power at $s_i$ under policy $\pi^s$ is $p_i(\pi^s)$, and thus $C_i^n$ and $\sum_{s=1}^{n} p_i(\pi^s)t_u$ forms $q^n_i$.
And we need to guarantee that $q^n_i\in [0,Q]$, where $Q$ is the capacity of the battery (Eq. (\ref{eq:rangeq_i})).
Eq. (\ref{eq:z_j}) shows that at each slot at most one policy can be selected, where $z^n_j$ denotes whether $\pi_j$ is selected in the $n$-th slot or not. 
Based on this formulation, the UMC problem is to determine the charging policies $\pi_j=\{l_j,\bold{w}_j\},~\forall j\in[|\Gamma|]$ for each slot,
{to maximize the utility of all sensors at the end of each round.}

{In the following, we first develop the candidate set of charging locations
$L$ to construct the policy set $\Gamma$.  
We then propose bandit-based approach to handle the target problem with the detailed regret analysis which guarantees the theoretical performance bounds. 
}

\section{Area Discretization}\label{sec:area}

To construct the set of stop locations ($L$) for the mobile charger, we divide the continuous 2D area into uniform grids with an edge length of $\epsilon$. 
{Each grid is a potential charging location, which is set up to simplify and discretize the space. The sensor's exact position is determined by its coordinates, meaning multiple sensors can be located within the same grid, or there may be no sensor at all in a grid.}
As illustrated in Section \ref{sec:system_model}, when considering the practical model of beamforming technique,
the power intensity is different from point to point even in the same grid, which brings significant challenges to the area discretization.
{\bf To address this challenge, we need to guarantee that the difference of power intensity inside the same grid can be bounded.}

In the following, by analyzing the gap between the minimum and maximum power intensity of a grid, given a charger with the fixed distance and fixed code word, we can evaluate the approximation ratio introduced by discretizing the 2D area.

\begin{figure}[t]
    \centering        \includegraphics[width=0.8\linewidth]{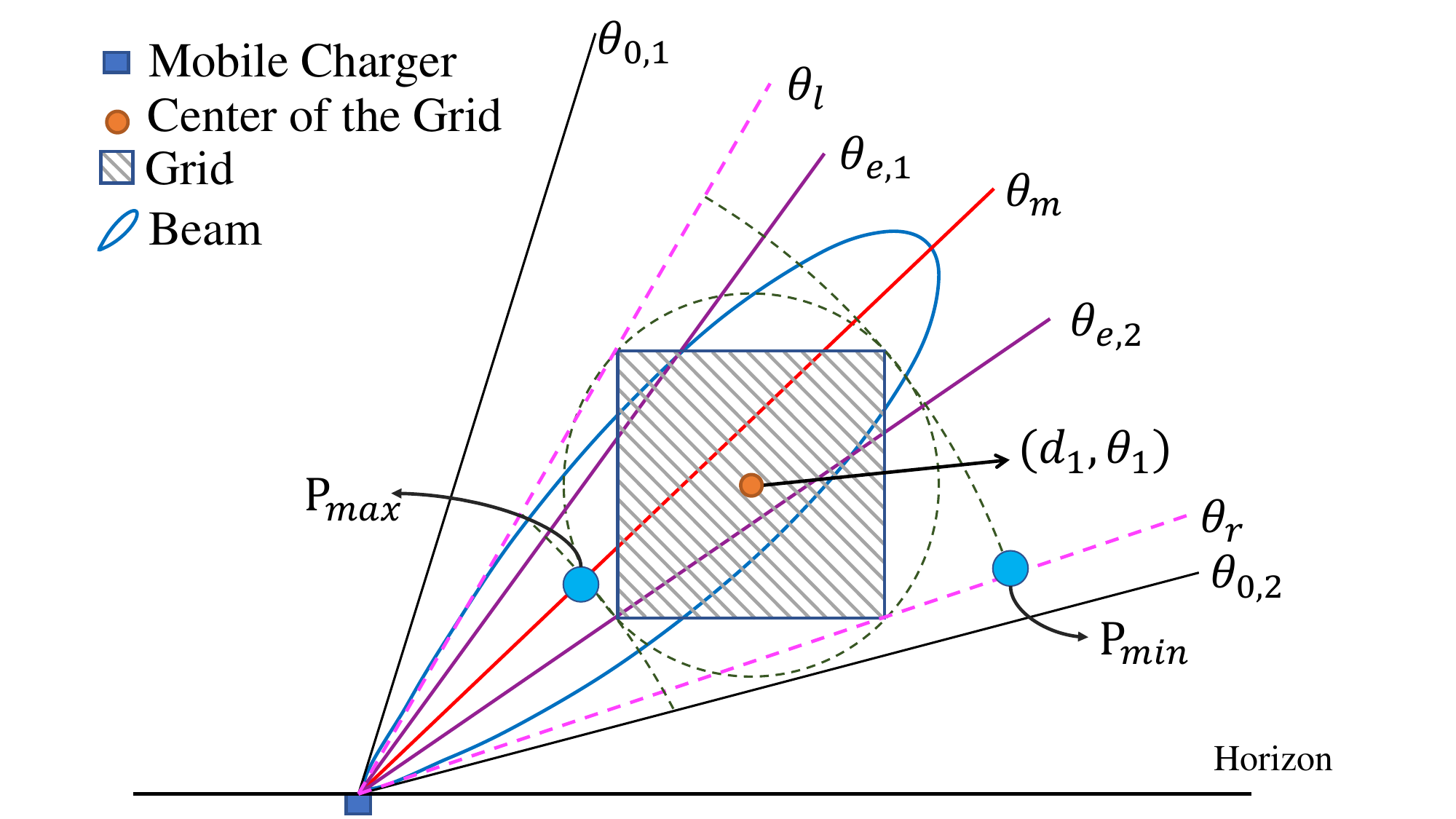}
    \caption{
    An example with beam from the charger directing to $\theta_m$ in polar coordinates covering the grid.
    \small {
    }
    }
    \label{fig:polar}
\end{figure}

As shown in Fig.~\ref{fig:polar}, considering the mobile charger as the center of the polar coordinates, the center of the grid is at the polar coordinate position $(d_1, \theta_1)$.
The main beam covers the sector of $[\theta_{0,1},\theta_{0,2}]$.
We use $P(d, \theta)$ to represent the power intensity at the polar coordinate position $(d, \theta)$, where $d$ is its distance to the charger, and $\theta$ is the angle formed by the horizon and the line which connects the charger and the position.
Based on the practical beamforming model, $P(d, \theta)$ satisfies the following three key features (see Fig.~\ref{fig:multi} and Fig.~\ref{fig:polar}):

\begin{enumerate}
    \item The power intensity in the main beam is symmetrical along $\theta_m$, where $\theta_m$ is the main direction of the beam with the maximum power intensity;
    \item Given the same $\theta$, the expectation of $P$ becomes larger along with the decrease of $d$, typically $P\sim \frac{1}{d^2}$ \cite{Balanis2005Antenna,choi2018theory};
    \item  Given the same $d$, $P$ is larger when $|\theta-\theta_m|$ decreases for $\theta\in [\theta_l,\theta_r]$, where $\theta_l$ and $\theta_r$ are the angles of left and right corners of the grid.
\end{enumerate}

Let $P_{max}$ and $P_{min}$ represent the upper bound of the maximum and the lower bound of the minimum power intensity in the grid with center at $(d_1,\theta_1)$, respectively.
{\bf The following theorem demonstrated that the area discretization brings constant approximation ratio $\alpha_0$ given the grid length $\epsilon$.}
\begin{theorem}
    \label{theorem:alpha}

    With a minimum distance between the mobile charger and a sensor\,\footnote{
    Let $d'$ represent the distance between the mobile charger and the sensor. Assume that $d' > \max( \frac{\sqrt{2}\epsilon}{2}\frac{sin(\pi/4-\theta_{0,1})}{sin(\theta_{e,1}-\theta_{0,1})}, \frac{\sqrt{2}\epsilon}{2}\frac{sin(\pi/4+\theta_{0,2})}{sin(\theta_{e,2}-\theta_{0,2})})$ so that the grid can be covered by the main beam.
    This is a reasonable assumption in real systems as there always exists a physical distance between sensor and the mobile charger. Meanwhile $d_1>\sqrt{2}\epsilon/2$ indicates the mobile charger and the sensor are not in the same grid.}, the approximation ratio of area discretization is constant:
\begin{equation}
    \alpha_0 \leq \frac{P_{max}}{P_{min}} = \gamma\frac{(d_1+\sqrt{2}\epsilon/2)^{2}}{(d_1-\sqrt{2}\epsilon/2)^{2}}
\end{equation}
where $d_1$ is the distance between the charger and the grid center. We use $\theta_{e,1}$ and $\theta_{e,2}$ to denote the two directions where the power intensity (with the same distance) is exactly $1/\gamma$ ($\gamma\in(1.25,2.5]$ for a reasonable ULA setting) of that along the direction of the main beam $\theta_m$.
\end{theorem}

\begin{proof}
    Let $P_{max}=P(d_{max},\theta_{max})$ and $P_{min}=P(d_{min},\theta_{min})$, 
    then we have $d_{max}=d_1+\frac{\sqrt{2}}{2}\epsilon$, $d_{min}=d_1-\frac{\sqrt{2}}{2}\epsilon$,
    $\theta_{max}=\mathop{\arg\min}\limits_{\theta}{(|\theta_m-\theta|)}$ and  $\theta_{min}=\mathop{\arg\max}\limits_{\theta}{(|\theta_m-\theta|)}$, where $\theta\in[\theta_l,\theta_r]$.

\noindent Considering $P\sim \frac{1}{d^2}$, we have
\begin{equation}
\begin{array}{ll}
    P_{max}&=P(d_{max},\theta_{max})
    \\  &= \frac{(d_1+\sqrt{2}\epsilon/2)^{2}}{(d_1-\sqrt{2}\epsilon/2)^{2}} P(d_{min},\theta_{max})
    \\  &= \gamma\frac{(d_1+\sqrt{2}\epsilon/2)^{2}}{(d_1-\sqrt{2}\epsilon/2)^{2}} P(d_{min},\theta_{min})
    \\  &= \gamma\frac{(d_1+\sqrt{2}\epsilon/2)^{2}}{(d_1-\sqrt{2}\epsilon/2)^{2}} P_{min}
\end{array}
\end{equation}
and
\begin{equation}
    \begin{array}{l}
    \alpha_0\leq \frac{P_{max}}{P_{min}}=\gamma\frac{(d_1+\sqrt{2}\epsilon/2)^{2}}{(d_1-\sqrt{2}\epsilon/2)^{2}}
    \end{array}
\end{equation}

Given $\gamma$ is a constant, $d_1>\frac{\sqrt{2}}{2}\epsilon$ and $\epsilon>0$, it is straightforward that $\alpha_0$ is constant-level.
In the following analysis, the approximation ratio of all the proposed schemes should times $\alpha_0$ which is introduced by the area discretization.
\end{proof}

\section{Bandit-based Scheme for UMC}\label{sec:scheme}

While the policy set $\Gamma=L \times \mathcal{M}$ ($\mathcal{M}$ is the codebook) has been developed, we now propose the bandit-based schemes to solve the UMC problem with unknown CSI and unknown energy consumption function for all sensors.
In the following, we first analyze the property of the oracle, i.e. when the CSI $\bold{h}_i$ and the total energy consumption during the charging round is known, and propose a greedy scheme GUA with constant ratio. This can help guide the design of the proposed two bandit algorithms, both with regret bound analysis.

\subsection{Property Analysis of the Oracle}
{
In oracle, the CSI $\bold{h}_i$ is stable and known at the beginning of each charging round.
To address the UMC problem under this scenario, below we propose a greedy procedure named GUA ({\bf G}reedy {\bf U}tility {\bf A}lgorithm) for selecting policy in a charging round, which is proved to achieve the constant-level approximation ratio. Let
\begin{equation}
    \label{eq:fcdot}
    f(\cdot)=\sum_{i=1}^{{N}}{U(q_i^{N_s})}-\sum_{i=1}^{{N}}{U(C_i^{N_s})}
\end{equation}
represent the objective function in the target UMC problem.
In each iteration of GUA, {the set of adopted policies} $\Pi$ greedily selects a policy $\pi_j \in \Gamma$ that can maximize the growth of its objective function $f(\cdot)$. 
Note that $\pi_j$ can be selected repeatedly and distinguished by the index of slot for each occurrence.

The time complexity of GUA is $O(NN_s|\Gamma| )$ for every charging round.
{\bf The following lemma and theorem show that GUA leads to a $(\frac{1}{2}-\frac{1}{2e})-$approximation ratio under some condition associated with the power consumption of sensors.}

\begin{lemma}
\label{lemma:mssfmp}
    The UMC problem is a monotone submodular set function maximization problem subject to the presented constraints.
\end{lemma}
The proof of Lemma \ref{lemma:mssfmp} is straightforward and thus omitted.
\qiu{In this context, the “set” refers to the selected charging actions over all time slots, represented by the binary variables ${z_j^n}$. The objective maps these decisions to the total utility. Since utility grows slower as a sensor's energy increases, charging low-energy sensors yields higher gains, showing diminishing returns. Thus, the objective is submodular. As charging never reduces utility, it is also monotone.}
This lemma will be used in the proof of Theorem \ref{theorem:offline}.

\begin{algorithm}[t]
\renewcommand{\thealgorithm}{GUA}
\caption{: Greedy Utility Algorithm}
\label{alg:greedy}
\begin{algorithmic}[1]
\REQUIRE Number of slots in a round $N_s$, energy information at the beginning of the round $X_t$, the expected received power $p_i(\pi_j), i\in[N], j\in [|\Gamma|]$.
\STATE $\Pi=\emptyset$
\FOR{$n=1,2,\cdots,N_s$}
\STATE $\pi'_n=\argmax\limits_{\pi_j\in \Gamma} {f(\Pi \cup \{\pi_j\})-f(\Pi)}$
\STATE Add ${\pi'_n}$ in $\Pi$ as the $n$-th element.
\ENDFOR
\end{algorithmic}
\end{algorithm}

\begin{theorem}
    \label{theorem:offline}
    When sensor $s_i$'s power consumption  $p^c_i=\frac{x_i}{\zeta T_c}$, {where $\zeta\ge1$ is a parameter} indicating that $s_i$ would consume $\frac{x_i}{\zeta}$ energy if not charged during the round, GUA achieves $(\frac{1}{2}-\frac{1}{2e})-$ approximation
, with either of the two conditions holds: 1) $\zeta \ge 2$; and 2) $U(\zeta p^c_i T_c)\le 2U((\zeta-1)p^c_i T_c)$ for $i\in [N]$.\footnote{The first condition implies that if sensor $s_i$ is not charged in this round, it has at least half of its initial energy remaining at the end of the round. \qiu{The second condition ensuresthe utility function does not increase too rapidly with energy. Intuitively, it enforces diminishing returns: adding more energy to a well-powered sensor should not yield disproportionately high utility gains.}}
\end{theorem}

\begin{proof}
    For an arbitrary instance, let $U^g_{\infty}$ denote the utility obtained by the greedy algorithm \ref{alg:greedy} and $OPT_{\infty}$ denote the optimum when $Q=\infty$.
    Relatively, for the same instance, let $U^g$ denote the utility obtained by the greedy algorithm \ref{alg:greedy} and $OPT$ denote the optimum when $Q\neq\infty$.

    Based on Lemma \ref{lemma:mssfmp}, it can be easily proved that \ref{alg:greedy} achieves $(1-\frac{1}{e})$-approximation to maximize $f(\cdot)$ when $Q=\infty$ and there is no polynomial-time algorithm that can achieve better approximation, which follows the proofs given in studies~\cite{chen2018contextual,fujishige2005submodular}. 
    Thus we have 
    \begin{equation}
    \label{eq:g'_OPT0}
        (1-1/e)OPT_{\infty}\le U^g\le OPT_{\infty}
    \end{equation}
    And apparently
    \begin{equation}
    \begin{array}{ll}
    \label{eq:g'_OPT}
         U^g\le OPT\le OPT_{\infty}, \\
         U^g\le U^g_{\infty}. \\
    \end{array}
    \end{equation}
    Considering the two situations that $s_i$ is fully charged for the first time in round $t$ in $n_1$-th slot with energy overflow and $s_i$ is fully charged for the last time in $n_2$-th slot without energy overflow, the wasted part of energy is at most $(n_2-n_1)p^c_i t_u\le p^c_i T_c $, where $s_i$ has a battery capacity of $Q$.
    Obviously, with Eq.~(\ref{eq:g'_OPT}) and given context $x_i$, we have 
    \begin{equation}
        \begin{array}{ll}
             U(Q) \ge U^g_{\infty} \ge U^g & \ge U(Q-p^c_i T_c) \\
             & \ge U(x_i-p^c_i T_c) \\ 
             & = U((\zeta-1)p^c_i T_c), 
        \end{array}
    \end{equation}
    where $x_i$ denotes the initial amount of stored energy of $s_i$ and $s_i\in (0,Q)$ for general analysis.
    Due to the submodularity of utility function, we have
    \begin{equation}
    \begin{array}{ll}
    \label{eq:bound_gap}
        U^g_{\infty}-U^g & \le U(Q)-U(Q-p^c_i T_c) 
                \\ & \le U(x_i)-U(x_i-p^c_i T_c)
                \\ & = U(\zeta p^c_i T_c)-U((\zeta-1)p^c_i T_c)
                \\ & \le U(p^c_i T_c)
                .
    \end{array}
    \end{equation}
    If $\zeta\ge 2$, i.e. the first condition holds as given in the theorem statement, according to the monotonicity of $U$, we have
    \begin{equation}
        U(p^c_i T_c)\le U((\zeta-1)p^c_i T_c) .
    \end{equation}   
    Then we can bound $U^g$ by
    \begin{equation}
        \begin{array}{ll}
         U^g_{\infty} &\le U(p^c_i T_c) + U^g \\
             &\le U((\zeta-1) p^c_i T_c ) + U^g\le 2U^g
        \end{array}
        \label{eq:uginfty}
    \end{equation}
    And we can finish the proof with $\zeta\ge 2$:
    \begin{equation}
        U^g \ge \frac{1}{2}U^g_{\infty} \ge \frac{1}{2}(1-\frac{1}{e})OPT_{\infty} \ge (\frac{1}{2}-\frac{1}{2e})OPT.
        \label{eq:ugapp}
    \end{equation}
    If $U(\zeta p^c_i T_c)\le 2U((\zeta-1)p^c_i T_c)$, i.e. the second condition holds, according to Eq.~(\ref{eq:bound_gap}), we have
    \begin{equation}
        U^g_{\infty}-U^g\le U((\zeta-1) p^c_i T_c )
    \end{equation} 
    Then we can bound $U^g$ by
    \begin{equation}
         U^g_{\infty} \le U((\zeta-1) p^c_i T_c ) + U^g\le 2U^g
        \label{eq:uginfty}
    \end{equation}
    And thus when $U(\zeta p^c_i T_c)\le 2U((\zeta-1)p^c_i T_c)$, we also have:
    \begin{equation}
        U^g \ge \frac{1}{2}U^g_{\infty} \ge \frac{1}{2}(1-\frac{1}{e})OPT_{\infty} \ge (\frac{1}{2}-\frac{1}{2e})OPT.
        \label{eq:ugapp}
    \end{equation}
\end{proof}

The proposed GUA and its theoretical result can be used to guide the design and analysis of the proposed bandit algorithms (see Line 4 in Algorithm \ref{alg:UMCB} and Eq. (\ref{eq:regret1})).

}
\subsection{Mapping to Bandit}

In the following, we first map the target UMC problem to the Multi-Arm Bandit problem, which is to make sequential decisions with unknown information based on the reward given by pulling arms in a round. With unknown CSI $\bold{h}_i$ and power consumption function $p_i^c$, we can try different policies by guessing and approaching the real $\bold{h}_i$ and $p_i^c$ round-by-round and update the policy based on the reward (the utility of sensors) obtained in last rounds.

Considering there are $T$ charging rounds, let $t$ represent the index of $T$ rounds.
CSI $\bold{h}_i$ follows some unknown distribution.
The harvested power of $s_i$ under $j$-th policy $\pi_j$ thus is regarded as a random variable of $p_i(\pi_j)$ with unknown expectation.  

\vspace{0.02in}
\noindent {\bf [Base arm $a_k$]} Each base arm represents a code word $\bold{w}_j$ at a specific location $l_j$ in the time slot $t_u$, corresponding to a charging policy $\pi_j=\{\bold{w}_j,l_j\}$ at each time slot $t_u$.
Given the policy set $\Gamma$ and $N_s$ time slots for charging, the set of all base arms is denoted as $\mathcal{K}=(a_1,\dots, a_k, \dots, a_K)$, where $K=|\Gamma| N_s$ represents the total number of available base arms
and $a_k$ denotes the $k$-th base arm corresponding to $\pi_j$ in the $n$-th time slot, where $k=(n-1)|\Gamma|+j$.
We define $p_i(a_k)$ as the expected harvested power of $s_i$ when employing $a_k$ {and normalize it to $[0,1]$}.
Note that among these $K=|\Gamma| N_s$ base arms, $a_p$ is identical to arm $a_q$ if the policy $\pi_j$ in these two arms are the same and the time slot indices are different, implying that we apply the same charging policy $\pi_j$ in two different time slots. Therefore, we have 
\begin{equation}
\label{eq:piak}
    p_i(a_k)=p_i(\pi_j,t_u)=p_i(\pi_j).
\end{equation}

 
\noindent {\bf [Super arm $S_t$]} We apply combinatorial bandit in this problem \cite{wang2017improving}, where a super arm (a combination of base arms) is selected in each round.
In each charging round $t$, we select a combination of $N_s$ base arms from all the $K$ base arms.
These selected base arms are combined to form a super arm $S_t\subseteq \mathcal{P}(\mathcal{K})$ and $|S_t|=N_s$, which serves as the charging policy for round $t$.



\vspace{0.02in}
\noindent {\bf [Reward $r_{S_t}$]}{
Given the formulation of utility function $U(\cdot)$ and the context information $X_t$\footnote{We assume that there are gaps between adjacent charging rounds and the initial energy states of sensors are at random levels at the start of each charging round. This follows the practical system settings that each sensor might also be powered by different energy harvesting techniques from the ambient sources. And each sensor also consumes energy according to the dynamic workload in the network.} in round $t$, the reward of the selected super arm $S_t$ can be calculated by observing the increase in utility from all sensors.
In round $t$, we maximize the reward of $S_t$, denoted as $r_{S_t}$. The objective can be represented as follows.
\begin{equation}
\max r_{S_t}=\mathbb{E}[\sum_{i=1}^{{N}}{U(q^{N_s}_i)}-\sum_{i=1}^{{N}}{U(C^{N_s}_i)} ].
\end{equation}

\subsection{UMCB Charging Schemes}

\begin{algorithm}[t]
\renewcommand{\thealgorithm}{UMCB}
\caption{: Utility Maximization Combinatorial Bandit}
\label{alg:UMCB}
\begin{algorithmic}[1]
\REQUIRE $T$
\FOR {t=1,2,3,\dots}
\STATE Observe the context $X_t$, update $T_{j,t}$ and $\hat{p}_{i}(\pi_j)$.
\STATE Update $\rho_{j}$: $\rho_{j,t}=\sqrt{\frac{3\ln{t}}{2T_{j,t}}}$.
\STATE Determine the super arm $S_t$ greedily based on the proposed GUA.
\STATE Play $S_t$ and observe rewards.
\ENDFOR
\end{algorithmic}
\end{algorithm}

Based on the above bandit setting, we present two upper confidence bound based algorithms: UMCB ({\bf U}itlity {\bf M}aximization {\bf C}ombinatorial {\bf B}andit) and UMCB-SW ({\bf S}liding {\bf W}indow) for stationary ($\bold{h}_i$ follows unknown but fixed distribution) and non-stationary ($\bold{h}_i$ follows unknown and dynamic distribution) CSI, respectively.
The performance of the algorithms are evaluated by the expected regret, which is the gap between the proposed algorithm and the oracle.{
These two algorithms are designed based on the framework of the classic bandit algorithm CUCB~\cite{wang2017improving}.
{\bf The regret of the proposed scheme UMCB is proved to be bounded in $O(ln T)$ with small coefficients.}
And it has been proved in \cite{kveton2015tight} that the distribution-dependent\footnote{Distribution-dependent bound refers to the specific characteristics of the performance boundary dependence on the distribution of arm rewards, such as gap to the optimal arm. Relatively, distribution-independent bound provides a performance guarantee established under all possible reward distribution. The former usually provides a tighter bound.} bound (given by any bandit framework) is lower bounded by $\Omega(log T)$, which indicates that {\bf the proposed UMCB achieves the tight regret bounds.}
}

\subsubsection{UMCB}

As shown in Eq. (\ref{eq:piak}), $p_i(\pi_j)$ denotes the expected harvested power of $s_i$ given policy $\pi_j\in \Gamma$.
Based on the charging history, $T_{j,t}$, the number of times $\pi_j$ is selected for $s_i$ until round $t$, and $\bar{p}_i (\pi_j)$, the empirical estimation of $p_i(\pi_j)$ for all sensors and policies, can be updated. 
Then, we can obtain the confidence radius $\rho_{j,t}$ of $p_i(\pi_j)$ as $\rho_{j,t} = \sqrt{\frac{3\ln{t}}{2T_{j,t}}}$\footnote{{The form of this confidence radius is derived from the Upper Confidence Bound (UCB) algorithm commonly used in multi-armed bandit problems.\cite{slivkins2019introduction}}} and thus obtain the UCB denoted as
\begin{equation}
   \hat{p}_i(\pi_j)=\bar{p}_i (\pi_j)+\rho_{j,t}. 
\end{equation}

After the initialization, we can consider the UCB by incorporating them into the maximization problem form of the submodular set function. 
Given $\hat{p}_i(\pi_j)$ of all sensors and policies, we apply the proposed scheme GUA to solve the general problem where $\bold{h}_i$ is unknown.
Starting from an empty set, $S_t$ greedily selects a base arm corresponding to policy $\pi^*\in\Gamma$ that can maximize the growth of the objective function $f(\cdot)$ as defined in Eq. (\ref{eq:fcdot}) in each step until it reaches $N_s$.
Then the charging strategy of round $t$ is determined. The details of the algorithm are shown in Algorithm \ref{alg:UMCB}.

The time complexity of UMCB is $O(NN_s|\Gamma|)$ for every charging round. It is worth noting that both UMCB and the following proposed UMCB-SW determine the charging policies ($\pi_j$ for all time slots) before the start of each charging round, so their computation time is not a concern for real-time charging. In practical systems, a charging round may last for hours, making the algorithm computation time negligible.   

{

Let $S^*_t$ denote the optimal strategy and $S_t$ is the UMCB policy.
Given $X_t$ in each round $t$, as the optimal solution is hard to obtain, we can measure the performance of UMCB by comparing to the proposed algorithm GUA with constant bound by defining the expected regret as
\begin{equation}
\label{eq:regret1}
    \min R_1(T)=\alpha \sum_{t=1}^{T} r_{S^*_t}-\sum_{t=1}^{T} r_{S_t},
\end{equation}
where $\alpha=\frac{1}{2}-\frac{1}{2e}$ is the approximation ratio of GUA.
}

We define the event $\mathcal{N}_t^s$ as: in round $t$, for arm $a_k, \forall k\in [K]$ and $s_i, \forall i\in [N]$, we have $|\hat{p}_{i,t-1}(a_k) - p_{i,t}(a_k)|\le \rho_{k,t}$, where $\rho_{k,t}=\sqrt{\frac{3\ln{T}}{2T_{k,t-1}}}$($\infty$ if $T_{k,t-1}=0$).
\qiu{Then we can have the following theorem, which show the regret bound of the proposed UMCB.}

\begin{lemma}
    For each round $t\ge 1$, $\Pr\{\neg \mathcal{N}_t^s\}\le 2NK t^{-2}$, where $K$ is the number of total base arms, and $N$ is the number of sensors.
\end{lemma}

\begin{proof}
{\small
    \begin{equation}
        \begin{array}{ll}
             \Pr\{\neg \mathcal{N}_t^s\}
             \\  = \Pr\{\exists k\in [K],\exists i\in [N] ,|\hat{p}_{i,t-1}(a_k) - p_i(a_k) |\ge \sqrt{\frac{3\ln{t}}{2T_{k,t-1}}}\}  
             \\  \le \sum_{i=1}^{N} \sum_{k\in [K]}   \Pr\{ |\hat{p}_{i,t-1}(a_k) - p_i(a_k) |\ge \sqrt{\frac{3\ln{t}}{2T_{k,t-1}}}\} 
             
            \\ \le \sum_{i=1}^{N} \sum _{k\in [K]} \sum_{s=1}^{N_s\cdot t} 
            \Pr\{|\hat{p}_{i,t-1}(a_k) - p_i(a_k) |\ge \sqrt{\frac{3\ln{t}}{2s}} \}
            \\ \overset{\textcircled{\scriptsize{1}}}{\le} \sum_{i=1}^{N}  \sum_{k\in [K]} \sum_{s=1}^{t} 2 t^{-3}
            \\  = 2NK t^{-2}.
        \end{array}
    \end{equation}
}
    The formula $\textcircled{\scriptsize{1}}$ is based on Hoeffding's Inequality.

\end{proof}

{

\begin{theorem}
\label{theorem:UMCB_bound}
    The expected regret for \ref{alg:UMCB} is
     $\mathbb{E}(R_1(T)) \le \frac{\pi^2}{3}NK \cdot \Delta_{max} + 2NB|\Gamma| + \sum_{j\in [|L|]}\frac{48B^2NN_s \ln{T}}{\Delta_{min}}$, where $\Delta_{max}=\max_{a_k\in \mathcal{K}}\Delta^{a_k}_{max}$, $\Delta_{min}=\min_{a_k\in \mathcal{K}}\Delta^{a_k}_{min}$, and $B$ is constant to describe the smoothness of $U(q_i)$ (See footnote 3).
     For $S \subseteq \mathcal{P}(\mathcal{K})$ and $|S|=N_s$, $\Delta^{a_k}_{max}=\sup\limits_{a_k\in S}{\max(0,\alpha\cdot r_{S^*_t}-r_{S_t})}$ and $\Delta^{a_k}_{min}=\inf\limits_{a_k\in S}{\min(0,\alpha\cdot r_{S^*_t}-r_{S_t})}$. If $\Delta^{a_k}_{max}$ and $\Delta^{a_k}_{min}$ are not positive, we define $\Delta^{a_k}_{max}=0$ and $\Delta^{a_k}_{min}=+\infty$ respectively.
\end{theorem} 



\begin{proof}
The oracle always succeeds with our scheme GUA.
We introduce a positive real number $M_k$ for $a_k$, where $M_S=\max_{a_k \in S} M_{a_k}$ for super arm $S$ and $M_{S}=0$ if $S=\emptyset$.
Now we define 
\begin{equation}
    \kappa_T(M,s)= \left\{ \begin{array}{ll}
2B & \textrm{if $s=0$}\\
2B \sqrt{\frac{6\ln{T}}{s}} & \textrm{if $1\le s\le l_T(M)$} \\
0 & \textrm{if $s > l_T(M)$},
\end{array} \right.
\label{eq:defofkappa}
\end{equation}
where
 \begin{equation}
\begin{array}{ll}
    l_T(M) & = \lfloor \frac{24B^2N_s^2 \ln{T}}{M^2} \rfloor.
\end{array}
\label{eq:ltm}
\end{equation}
We use $r_{S_t}(\bf p)$ to denote the reward with expectation vector ${\bf p}$ of arms, where ${\bf p}=\{\dots,p_{i,t}(a_k),\dots\}$ and ${\bf \Bar{p}}=\{\dots,\Bar{p}_{i,t}(a_k),\dots\}$. By $\mathcal{N}_t^s$, we have
\begin{equation}
    r_{S_t}(\Bar{\bf p})\ge \alpha \cdot r(S_t^*) \ge r_{S_t}({\bf p})+\Delta_{S_t},
\end{equation}
and 
\begin{equation}
\begin{array}{ll}
    &\Bar{p}_{i,t}(a_k)- p_{i,t}(a_k) - \frac{M_k}{2BN_s} \\ 
    \le & \min \{2\rho_{i,t}(a_k),1\} - \frac{M_k}{2BN_s} \\
    \le & \min \{2\sqrt{\frac{3\ln{T}}{2T_{k,t-1}}},1\} - \frac{M_k}{2BN_s}.
\end{array}
\end{equation}
By Lipschitz continuity condition, when $\Delta_{S_t}\ge M_{S_t}$, we have
\begin{equation}
\begin{array}{rl}
    \Delta_{S_t} & \le r_{S_t}(\Bar{{\bf p}}) -  r_{S_t}({\bf p}) = \sum_{i=1}^N (U( \Bar{q}_i^{N_s}) - U( q_i^{N_s}) ) 
    \\ & =  B \sum_{i=1}^N \sum_{a_k \in S_t} (\Bar{p}_{i,t}(a_k)- p_{i,t}(a_k) ) 
    \\ & \le \sum_{i=1}^N (2B \sum_{a_k \in S_t} (\Bar{p}_{i,t}(a_k)- p_{i,t}(a_k) ) - M_{S_t} )
    \\ & = \sum_{i=1}^N\sum_{a_k \in S_t} 2B (\Bar{p}_{i,t}(a_k)- p_{i,t}(a_k) - \frac{M_k}{2B|S_t|} ) 
    \\ & = \sum_{i=1}^N\sum_{a_k \in S_t} 2B (\Bar{p}_{i,t}(a_k)- p_{i,t}(a_k) - \frac{M_k}{2BN_s} ) 
    \\ & \le \sum_{i=1}^N \sum_{a_k \in S_t} \kappa_T(M_k,T_{k,t-1})
    \\ & = N \sum_{a_k \in S_t} \kappa_T(M_k,T_{k,t-1}).
\end{array}
\end{equation}
Then we can combine regret in each round when both of ${\Delta_{S_t}\ge M_{S_t}}$ and $\mathcal{N}_t^s$ happen:
\begin{equation}
\begin{array}{ll}
    &\sum_{t=1}^{T} \mathbb{I}(\{\Delta_{S_t}\ge M_{S_t}\})\wedge \mathcal{N}_t^s)\cdot \Delta_{S_t} 
    \\ \le &  N\sum_{t=1}^{T} \sum_{a_k \in \bar{S_t}} \kappa_T(M_k,T_{k,t-1})
    \\ = & N\sum_{k\in [K]}\sum_{s=0}^{T_{k,T}}\kappa_T(M_k,s)
    \\ = & N\sum_{j\in [|\Gamma|]}\sum_{s=0}^{T_{j,T}-N_s}\kappa_T(M_{j},s)
    \\ \le & N\sum_{j\in [|\Gamma|]}\sum_{s=0}^{l_T({M_k})}\kappa_T(M_{j},s)
    \\ = & 2NB|\Gamma| + N\sum_{j\in [|\Gamma|]}\sum_{s=1}^{l_T(M_{j})}\kappa_T(M_{j},s),
\end{array}
\end{equation}
where index $j$ incorporates all base arms with same policy $\pi$.
\begin{equation}
\begin{array}{ll}
    &\sum_{t=1}^{T} \mathbb{I}(\{\Delta_{S_t}\ge M_{S_t}\})\wedge \mathcal{N}_t^s)\cdot \Delta_{S_t} 
    \\ \le & 2NB|\Gamma| + N\sum_{j\in [|\Gamma|]}\sum_{s=1}^{l_T(M_{j})}\kappa_T(M_{j},s).
\end{array}
\label{eq:34}
\end{equation}
When a policy is adopted, the feedback of this slot indicates its charging efficiency and thus its CSI in this charging location can be obtained.
With the known CSI, the rewards of some of the unadopted policies (with the same location, but with different code words) can also be obtained.
These policies sharing the same location can thus be considered as explored in this slot.
In that case, the number of enumeration $|\Gamma|$ for $j$ could be reduced to $|L|$. And Eq. (\ref{eq:34}) can be rewritten as follows (formula \textcircled{\scriptsize{1}} of Eq.~(\ref{eq:longeq})).
\begin{equation}
\begin{array}{ll}
    &\sum_{t=1}^{T} \mathbb{I}(\{\Delta_{S_t}\ge M_{S_t}\})\wedge \mathcal{N}_t^s)\cdot \Delta_{S_t} 
    \\  \overset{\textcircled{\scriptsize{1}}}{\le} & 2NB|\Gamma| + N\sum_{l\in [|L|]}\sum_{s=1}^{l_T(M_{l})}\kappa_T(M_{l},s)
    \\   \overset{\textcircled{\scriptsize{2}}}{\le} & 2NB|\Gamma| + 2NB\sum_{l\in [|L|]}\sum_{s=1}^{l_T(M_{l})}\sqrt{\frac{6\ln{T}}{s}}
    \\  \le & 2NB|\Gamma| + 2NB\sum_{l\in [|L|]}\int_{s=0}^{l_T(M_{l})}\sqrt{\frac{6\ln{T}}{s}}\mathrm{d}s
    \\   \le & 2NB|\Gamma| + 4NB\sum_{l\in [|L|]}\sqrt{6\ln{T}\frac{24B^2N_s^2\ln{T}}{M_l^2}}
    \\  \le & 2NB|\Gamma| + \sum_{l\in [|L|]}\frac{48B^2NN_s \ln{T}}{M_l},
\end{array}
\label{eq:longeq}
\end{equation}
where $l$ is the index of all the available locations. Formula \textcircled{\scriptsize{2}} is based on Eq.~(\ref{eq:defofkappa}).
Due to the Lemma $\Pr\{\neg \mathcal{N}_t^s\}\le 2NK t^{-2}$. Let $R_2(t)$ denote $R_1(t) -R_1(t-1)$ for $t>=1$.
\begin{equation}
\begin{array}{ll}
    R_2(t) &\le R_2(t)(\neg\mathcal{N}_t^s) + R_2(t)(\Delta_{S_t}<M_{S_t}) 
    \\ &\quad + R_2(t)(\{\Delta_{S_t}\ge M_{S_t}\} \wedge \mathcal{N}_t^s))
    \\    &\le 2NK t^{-2} \cdot \Delta_{max}  + R_2(t)(\Delta_{S_t}<M_{S_t})
    \\ &\quad + R_2(t)(\{\Delta_{S_t}\ge M_{S_t}\} \wedge \mathcal{N}_t^s)).
\end{array}
\end{equation}
Take $M_j=\Delta_{min}$, then we have $R_2(t)\{\Delta_{S_t}<M_{S_t}\}=0$ and we can bound $R_1(T) =\sum_{t=1}^T R_2(t) $ as
\begin{equation}
\begin{array}{ll}
    R_1(T) \le \frac{\pi^2}{3}NK\cdot \Delta_{max} + 2NB|\Gamma| + \sum_{l\in [|L|]}\frac{48B^2NN_s \ln{T}}{\Delta_{min}}.
\end{array}
\end{equation}

\end{proof}

}

\subsubsection{UMCB-SW}

In real systems, the channel state distribution often varies along the time due to the dynamics of the environment.
The reference significance of earlier rounds will gradually decrease, while the reference significance of the most recent rounds of data information can be guaranteed.  Based on this intuition, UMCB is still applicable by introducing sliding window to capture the recent CSI. 
The only change in UMCB is to define and update $T_{j,t}$ referring to the number of times that policy $\pi_j$ has been adopted in the last $w$ rounds, where $w\leq T$ is the sliding window size.




{
Over the time interval $I=[t_1, t_2]$, the number of changes in channel state distribution, which is equivalent to the number of changes in the harvested power distribution at sensors, is denoted as 
{$D_I := 1+\sum_{t=t_1+1}^{t_2} \mathbb{I}(p_{i,t} (a_k ) \neq p_{i,t-1} (a_k )), i \in [{N}], k \in [K]$.}
The total change of the mean is defined as $V_I:=  \sum_{t=t_1+1}^{t_2} \lVert \bf{p}_{t}-\bf{p}_{t-1} \rVert_{\infty}$, {where $\bold{p}_{t}=\{p_1(a_1),\dots,p_i(a_k),\cdots,p_N(a_{K})\}$}.
For simplicity, $D$ and $V$ are used to denote $D_{[1,T]}$ and $V_{[1,T]}$ respectively.
}

\begin{theorem}
\label{theorem:UMCBSW_bound}
  Given $V$, the expected regret $\mathbb{E}(R_1(T)) \le \frac{\pi^2}{3}NK\cdot \Delta_{max} +(3-\frac{1}{e})KNB\sqrt{VT}  +2NB\sum_{l\in[|L|]} (\sqrt{VT}+1) (\sqrt{6\ln{T}} + \frac{24BN_s \ln{T}}{\Delta_{min}})$ when choosing $w=\min \{\sqrt{\frac{T}{V}},T\}$.
\end{theorem}

\begin{proof}
$\kappa_T$ is defined as follows:
\begin{equation}
    \kappa_T(M,s)= \left\{ \begin{array}{ll}
2B\sqrt{6\ln{T}} & \textrm{if $s=0$}\\
2B \sqrt{\frac{6\ln{T}}{s}} & \textrm{if $1\le s\le l_T(M)$} \\
0 & \textrm{if $s > l_T(M)$},
\end{array} \right.
\end{equation}
where $l_T(M)$ is defined in Eq.~(\ref{eq:ltm}).
The time $\{1,2,...,T\}$ is divided into $\Omega\le \lceil \frac{T}{w} \rceil $ segments $[1=t_0+1,w=t_1],...,[t_{\Omega-1}+1,t_{\Omega}=T]$, where each segment has length $w$, except for the last segment.
We use $T'_{k,t}$ to denote the number of times that arm $a_k$ is pulled in $[t_{n-1}+1,t_n-1]$ for every $k,t$ and $t_{n-1} < t \le t_n$.
Besides, we define $\nu_{k,t}$ as $\nu_t=\{\dots,\nu_{i,t}(a_k),\dots\}$ and $\nu_{i,t}(a_k)=  \frac{1}{T_{k,t}} \sum^{t-1}_{s=t-w+1} \mathbb{I} \{a_k \text{ is pulled in round } s\} p_{i,t}(a_k)$.
\begin{equation}
    \sum_{t=1}^T \mathbb{I}(a_k \in S_t)\cdot \kappa_T(M_k,T_{k,t})\le \sum_{t=1}^T \mathbb{I}(a_k \in S_t)\cdot \kappa_T(M_k,T'_{k,t})
\end{equation}
We define $\mathcal{N}_t^s$ as: we have $|\hat{p}_{i,t-1}(a_k) - \nu_{i,t}(a_k)|\le \rho_{k,t}$.
We bound $\sum_{t=1}^T \mathbb{I}(a_k \in S_t)\cdot \kappa_T(M_k,T_{k,t})$ as:
\begin{equation}
\begin{array}{ll}
    &\sum_{t=1}^T \mathbb{I}(a_k \in S_t)\cdot \kappa_T(M_k,T_{k,t}) 
    \\ & \le \sum_{t=1}^T \mathbb{I}(a_k \in S_t)\cdot \kappa_T(M_k,T'_{k,t})
    \\ & =   \sum_{n=1}^\Omega \sum^{t_{n}}_{t=t_{n-1} +1} \mathbb{I}(a_k \in S_t)\cdot \kappa_T(M_k,T'_{k,t})
    \\ & \le \sum_{n=1}^\Omega \sum_{s=0}^{w-1} \kappa_T(M_k,s)
    \\ & \le \sum_{n=1}^\Omega (2B\sqrt{6\ln{T}} + \sum_{s=1}^{l_T({M_k})}\kappa_T(M_{k},s))
    
\end{array}
\end{equation}
Similar to proof of Theorem~\ref{theorem:UMCB_bound}, it can be rewritten as 
\begin{equation}
\begin{array}{ll}
    \label{eq:le4}
    &\sum_{t=1}^T \mathbb{I}(a_k \in S_t)\cdot \kappa_T(M_k,T_{k,t}) 
    \\ & \le \sum_{n=1}^\Omega (2B\sqrt{6\ln{T}} + \sum_{s=1}^{l_T({M_l})}\kappa_T(M_{l},s))
    \\ & \le (\frac{T}{w}+1) (2B\sqrt{6\ln{T}} + \frac{48B^2N_s \ln{T}}{M_l}) 
\end{array}
\end{equation}
From Lemma 5 in \cite{chen2021combinatorial}, when $\mathcal{N}_t^s$ happens, we have 
\begin{equation}
    r_{S_t}(\Bar{p}_t)\ge r_{S_t}(\nu_{t})+\Delta_{S_t}^t-\sum_{i=1}^N(1+\alpha)KB\sum_{s=t-w+2}^t || {\bf p}_s - {\bf p}_{s-1} ||_\infty.
\end{equation}
Note $ \alpha=\frac{1}{2}-\frac{1}{2e} $. If ${\Delta_{S_t}^t\ge M_{S_t}}$, similar to proof of Theorem~\ref{theorem:UMCB_bound}, we have
\begin{equation}
    \begin{array}{ll}
        \Delta_{S_t}^t & \le r_{S_t}(\Bar{\bf p})-r_{S_t}(\nu_t)
            \\& \quad +\sum_{i=1}^N(1+\alpha)KB\sum_{s=t-w+2}^t || {\bf p}_s - {\bf p}_{s-1} ||_\infty
         \\ &  \le \sum_{i=1}^N(2B\sum_{a_k\in S_t}(p_{i,t}(a_k)-\nu_{i,t}(a_k))-M_{S_t}
            \\&\quad +2(1+\alpha)KB\sum_{s=t-w+2}^t || {\bf p}_s - {\bf p}_{s-1} ||_\infty )
         \\ &  \le 2NB\sum_{a_k\in S_t}(p_{i,t}(a_k)-\nu_{i,t}(a_k)-\frac{M_k}{2BN_s})
            \\&\quad +(3-\frac{1}{e})KNB\sum_{s=t-w+2}^t || {\bf p}_s - {\bf p}_{s-1} ||_\infty
         \\ &  \le N\sum_{a_k\in S_t} \kappa_T(M_k,T_{k,t}) 
            \\&\quad +(3-\frac{1}{e})KNB\sum_{s=t-w+2}^t || {\bf p}_s - {\bf p}_{s-1} ||_\infty         
    \end{array}
\end{equation}
Then with Eq.~(\ref{eq:le4}) we have 
\begin{equation}
    \begin{array}{ll}
        & R_1(T)(\{\Delta_{S_t}\ge M_{S_t}\} \wedge \mathcal{N}_t^s)) 
        \\&  \le N\sum_{t=1}^T \mathbb{I}(a_k \in S_t)\cdot \kappa_T(M_k,T_{k,t}) 
            \\&\quad +(3-\frac{1}{e})KNBw\sum_{s=2}^t || {\bf p}_s - {\bf p}_{s-1} ||_\infty
         \\ &  \le N\sum_{l\in[|L|]} (\frac{T}{w}+1) (2B\sqrt{6\ln{T}} + \frac{48B^2N_s \ln{T}}{M_l})
            \\&\quad +  (3-\frac{1}{e})KNBVw
         
    \end{array}
\end{equation}

Take $M_j=\Delta_{min}$ and $w=\min \{\sqrt{\frac{T}{V}},T\}$, then we have $R_1(t)\{\Delta_{S_t}<M_{S_t}\}=0$ and 
\begin{equation}
\begin{array}{ll}
    & R_1(T) \le R_1(T)(\neg\mathcal{N}_t^s) + R_1(T)(\Delta_{S_t}<M_{S_t})
        \\&\quad +R_1(T)(\{\Delta_{S_t}\ge M_{S_t}\} \wedge \mathcal{N}_t^s))
    \\&\le \frac{\pi^2}{3}NK\cdot \Delta_{max} +  (3-\frac{1}{e})KNBVw 
    \\&\quad + N\sum_{l\in[|L|]} (\frac{T}{w}+1) (2B\sqrt{6\ln{T}} + \frac{48B^2N_s \ln{T}}{M_l})
    \\&\le \frac{\pi^2}{3}NK\cdot \Delta_{max} +(3-\frac{1}{e})KNB\sqrt{VT} 
    \\&\quad +2NB\sum_{l\in[|L|]} (\sqrt{VT}+1) (\sqrt{6\ln{T}} + \frac{24BN_s \ln{T}}{\Delta_{min}})  
\end{array}
\end{equation}
\end{proof}
 
\subsection{Discussion on Multiple Mobile Chargers}
\label{sec:discussion}

We now discuss a more general scenario where multiple mobile chargers are deployed for charging.
In practical systems, multiple beams can enhance (or weaken) beam directionality and thus lead to more (or less) energy transfer due to the increased interference phenomenon of the superimposed waves.
\emph{When multiple mobile chargers simultaneously transmit wireless energy to one sensor, the harvested power of this sensor is thus not equivalent to the sum of the received powers, but many existing studies assumed that the received power from different chargers is cumulative\cite{dai2018wireless,dai2021placing,lin2023maximizing}.}


Fig.~\ref{fig:superimposed_waves} shows the simulated result of two orthogonal chargers applied with the same code word in a $3m\times3m$ area, wherein the chargers align their beams towards the center of this area to facilitate charging.
Respectively, Fig.~\ref{fig:multi1}, Fig.~\ref{fig:multi2} and Fig.~\ref{fig:multi3} show the power intensity maps of a single charger charging from the left-hand side, a single charger charging from the upside, and two orthogonal chargers charging simultaneously from left-hand side and upside.
It can be observed from Fig.~\ref{fig:multi3} that the orthogonal wave superposition enhances the power intensity in most of the overlapped area, while in the meantime weakens the power intensity in some certain area, which demonstrated that the power harvested from two different charging waves are not cumulative.


When multiple mobile chargers simultaneously transmit wireless energy to the same sensor, the received power must be recalculated. The proposed algorithms UMCB and UMCB-SW can still be applied.
The difference lies in defining the base arm as the joint strategy of these mobile chargers, represented as $\pi_j = (\pi_{1,j_1},..., \pi_{Nc,j_{Nc}})$, where $Nc$ denotes the number of mobile chargers.
Although $|\Gamma|$ will expand based on the number of chargers due to the principle of combinations, this only occurs in situations where the beams overlap with each other. When the charging beams do not interfere with each other, the growth of $|\Gamma|$ is linear.
The reward of each benchmark arm can indicate the energy status when multiple mobile chargers operate concurrently. 
{\bf It can be easily demonstrated that with multiple mobile chargers, the regret bounds of the algorithms remain unchanged, demonstrating efficient adaptability of UMCB(-SW).}


\begin{figure}[t]
\centering
    
    \subfigure{
    \includegraphics[width=0.31\linewidth]{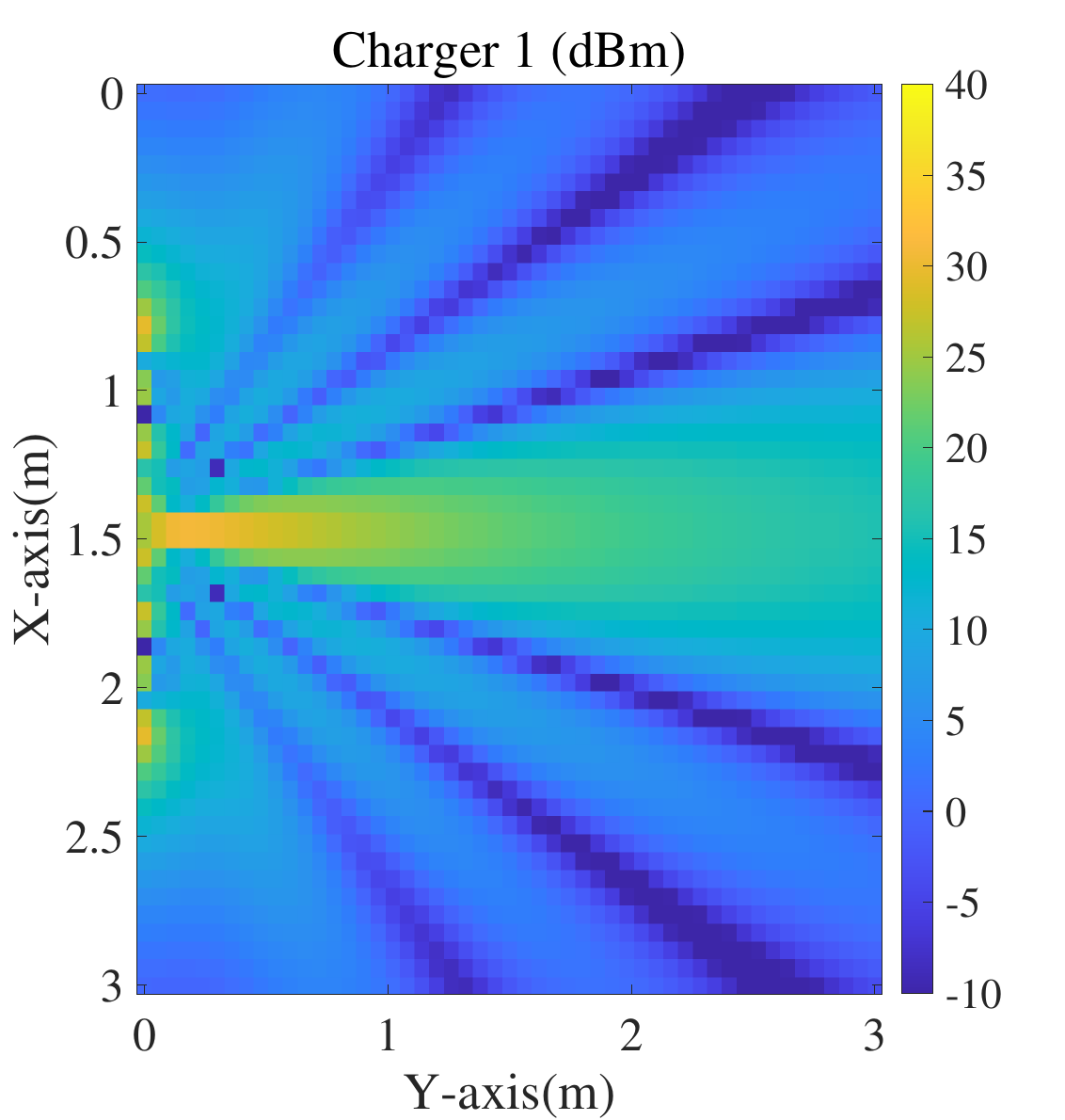}
    \label{fig:multi1}
    }\hspace{-0.7em}
    \subfigure{
    \includegraphics[width=0.31\linewidth]{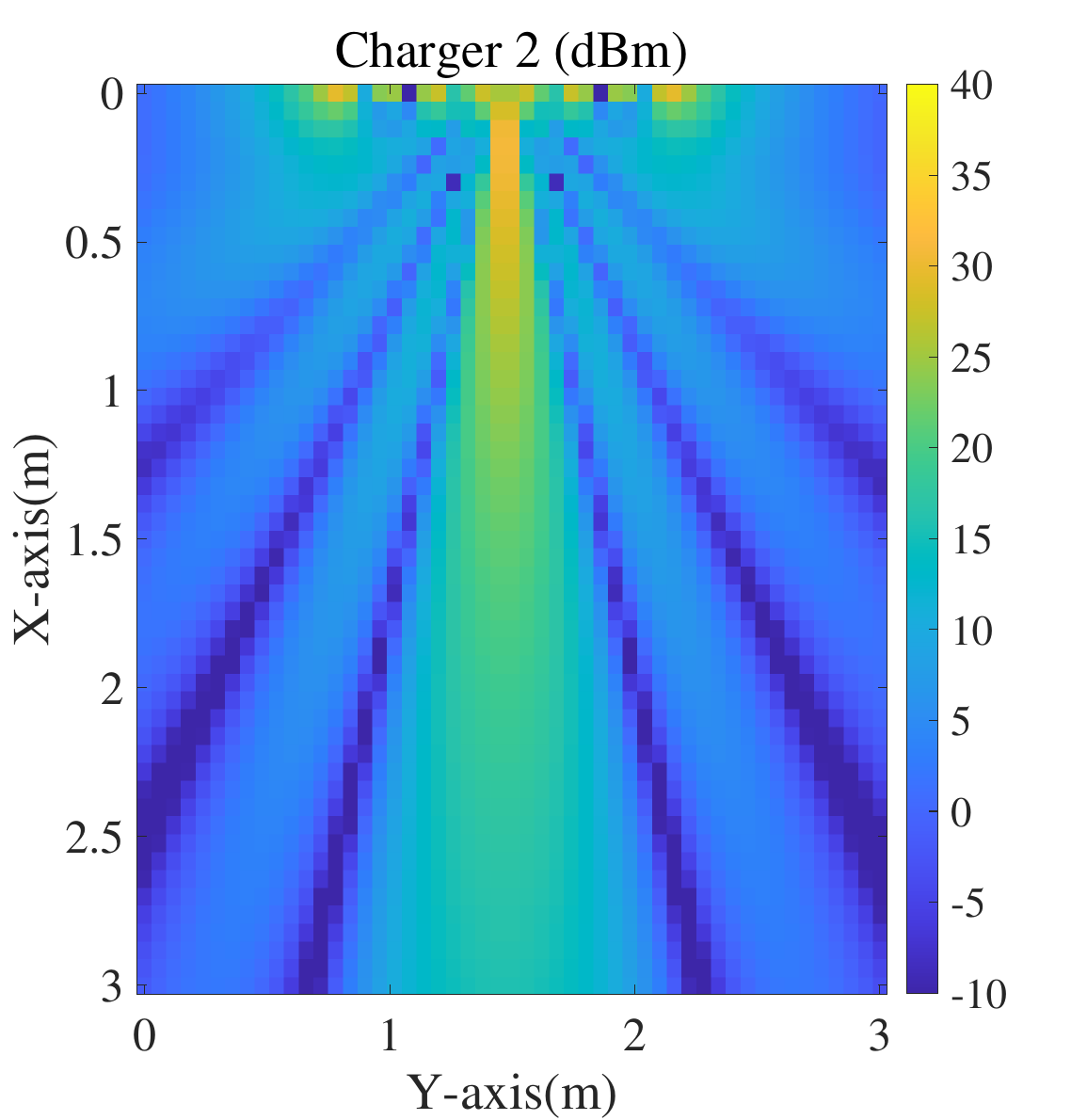}
    \label{fig:multi2}
    }\hspace{-0.7em}
    \subfigure{
    \includegraphics[width=0.31\linewidth]{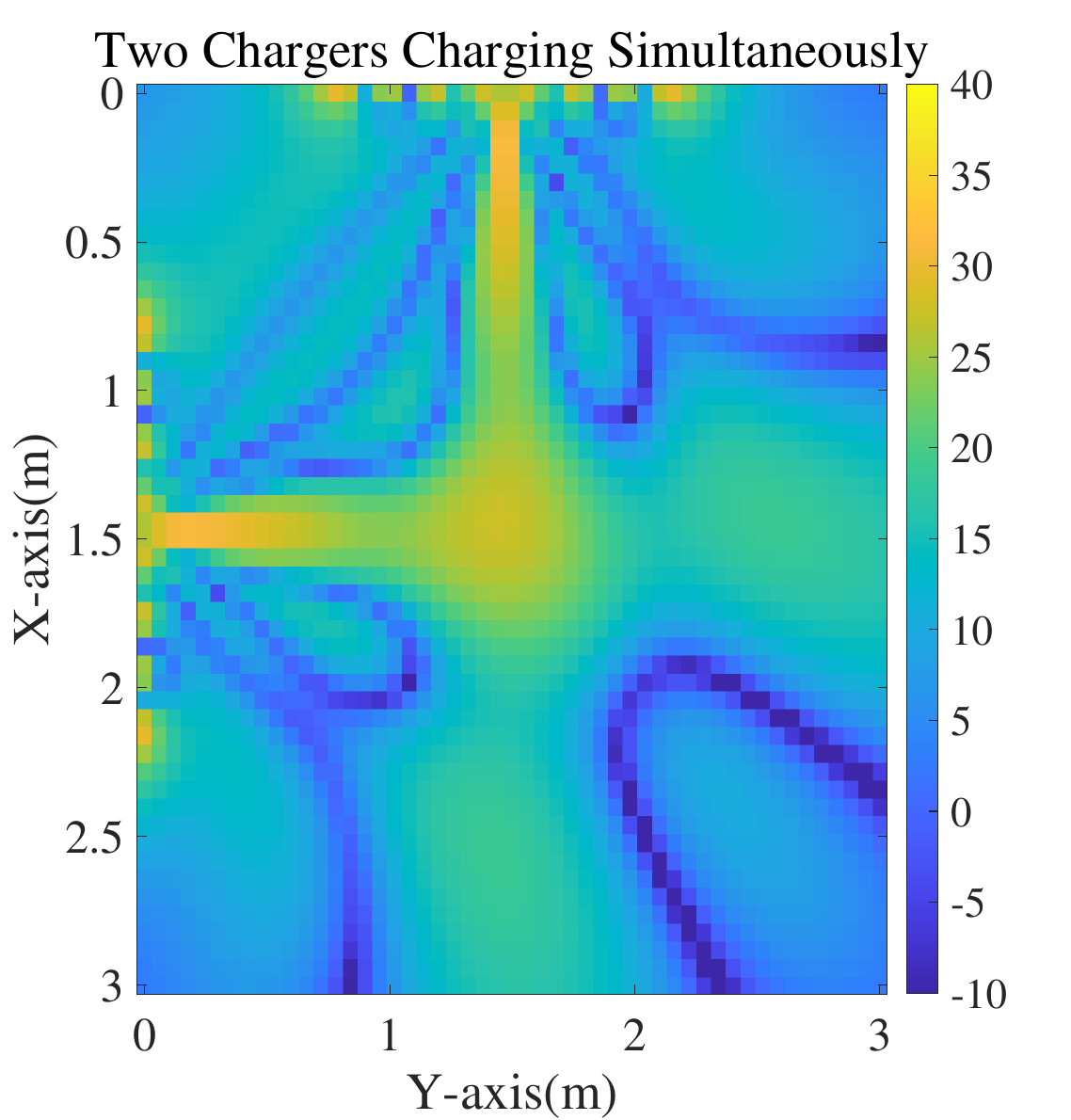}
    \label{fig:multi3}
    }
    \caption{
    Received power from 2 orthogonal chargers with the same code word, simulated with the perfect beamforming vectors through the 800MHz wave from 8-antenna ULA with a spacing of 0.1m.}
    \label{fig:superimposed_waves}
\end{figure}

\section{Performance Evaluation}\label{sec:experiments}

We conduct two experiments to evaluate the performance of \ref{alg:greedy} and the proposed bandit schemes.

\subsection{Experimental Settings}
\label{sec:exp_set}
In the experiments, we assume that sensors are randomly distributed within a $20m\times20m$ square area. This area is further discretized into $0.5m \times 0.5m$ grids to represent available positions. The size of the codebook is $M=12$, with code words uniformly distributed across angles. Each experiment includes $N$ sensors. {\bf We consider two utility function formulations in the experiments: $U_1(\sigma)= \frac{100}{N}\sqrt{\sigma}$ and $U_2(\sigma)=\frac{100}{N}\sqrt{\frac{\sigma}{1+\sigma}}$, where $\sigma=\frac{q_i}{Q}$ and $q_i$ is the energy in battery of $s_i$.} 
\eat{\begin{itemize}
    \item   $U_1(\sigma)= \frac{100}{N}\sqrt{\sigma}$ 
    \item   $U_2(\sigma)=\frac{100}{N}\sqrt{\frac{\sigma}{1+\sigma}}$.
\end{itemize}}
W.l.o.g. we set $N_s=1000$, $t_u=1$ and $T_c=1000$.
Before each round, for each $s_i$, its residual energy is randomly generated according to the uniform distribution in $(0,0.3Q]$ with $\zeta=2$ introduced in Theorem~\ref{theorem:offline} ($0.3$ can also be any other numbers within $(0,1)$ without affecting the experimental results).  The detailed settings for both offline and online scenarios are presented below. 


\subsubsection{Evaluation of Oracles}
In this set of experiments, the number of sensors $N$ is set to $30$. We conducted $100$ rounds of simulations, where results of each round is independent. We compare \ref{alg:greedy} with the following strategies:
\begin{itemize}
    \item {\bf Upper Bound (UB)}: UB is the oracle of an ideal scenario of UMC problem where the time is continuous and the battery capacity $Q$ is unlimited. The details of UB analysis is provided in Section \ref{sec:exp_ub}.
    \item {\bf Greedy MaxQ (GMQ)}:
    {this is a greedy scheme that in each time slot, it selects the policy that maximizes the received energy (while GUA selects the policy that maximizes the growth of the objective function).}
    \item {\bf $(1/2-1/2e) \times$ Upper Bound (AUB)}: this is the theoretical upper bound of the approximation ratio for the proposed GUA as proved in Theorem \ref{theorem:offline}.
\end{itemize} 

\begin{figure}[t]
    \centering
    \subfigure[With $U_1$]{
        \includegraphics[width=0.465\linewidth]{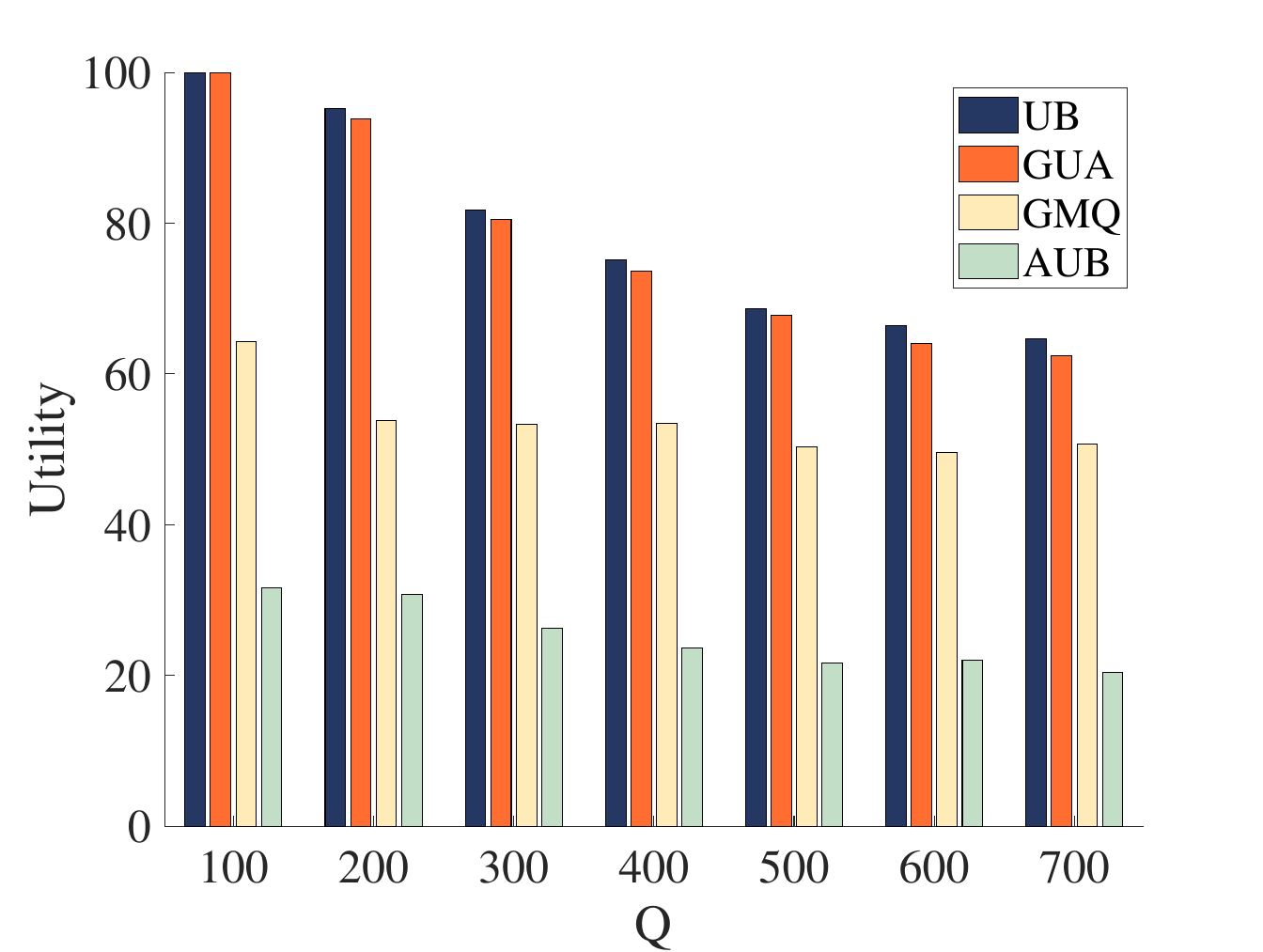}
        \label{fig:off_o1}
    }
    \subfigure[With $U_2$]{
        \includegraphics[width=0.465\linewidth]{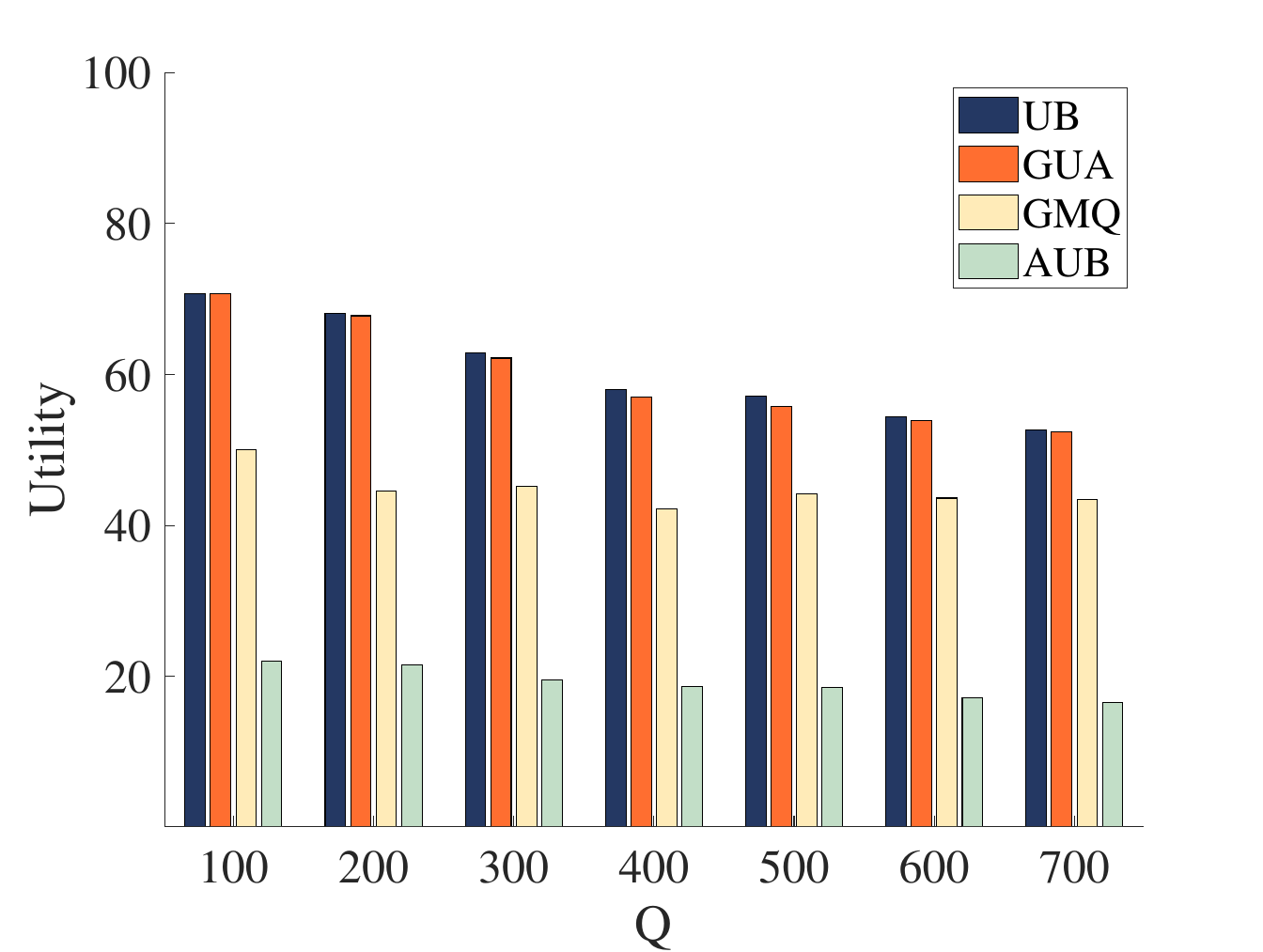}
        \label{fig:off_o3}
    }
    \caption{Average utility with varied $Q$ for oracles}
    \label{fig:off_q}
\end{figure}

\begin{figure*}[t]
    \centering
    \subfigure[With $U_1$]{
        \includegraphics[width=0.235\linewidth]{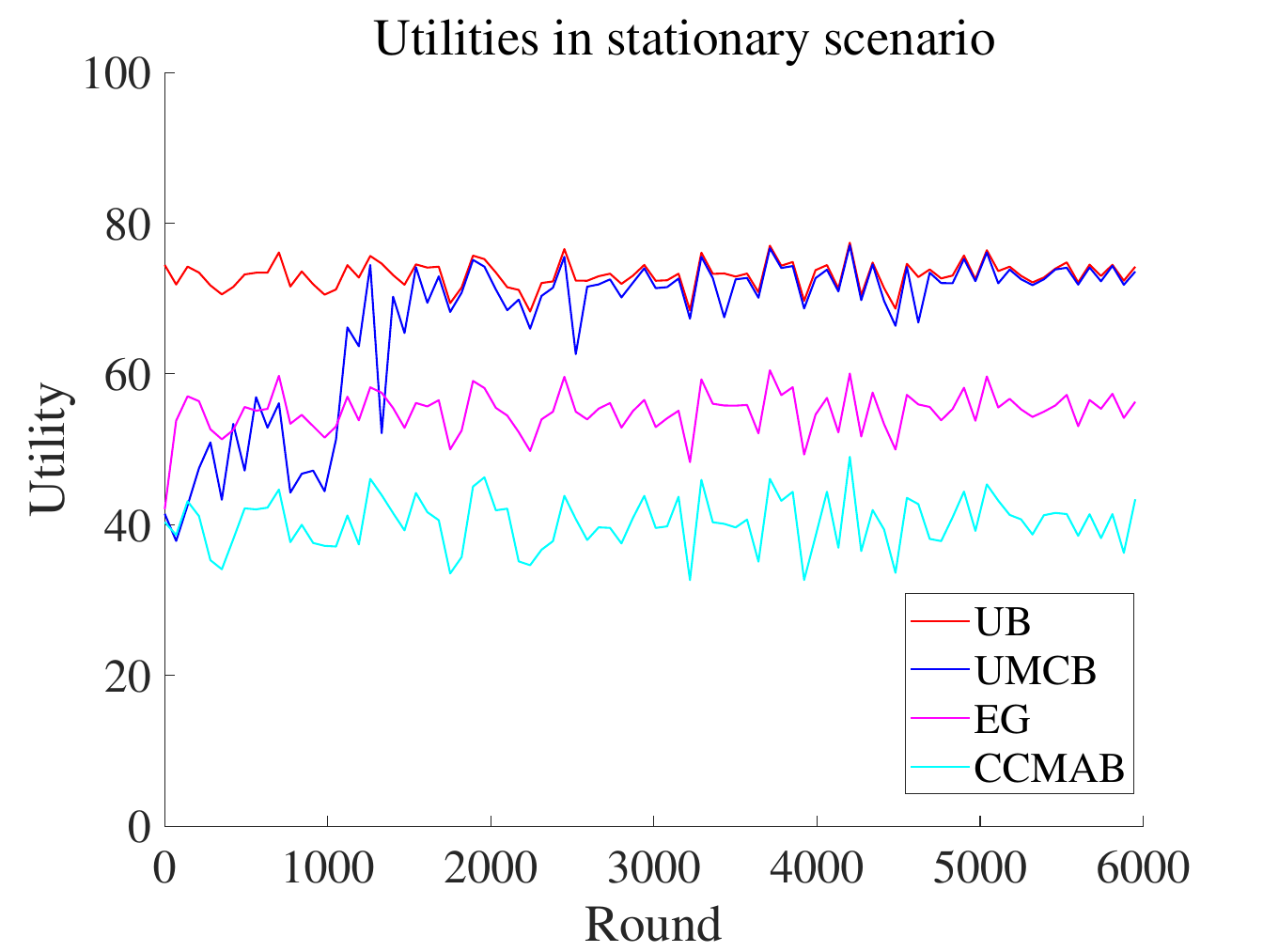}
        \label{fig:s_o1_res}
    
        \includegraphics[width=0.235\linewidth]{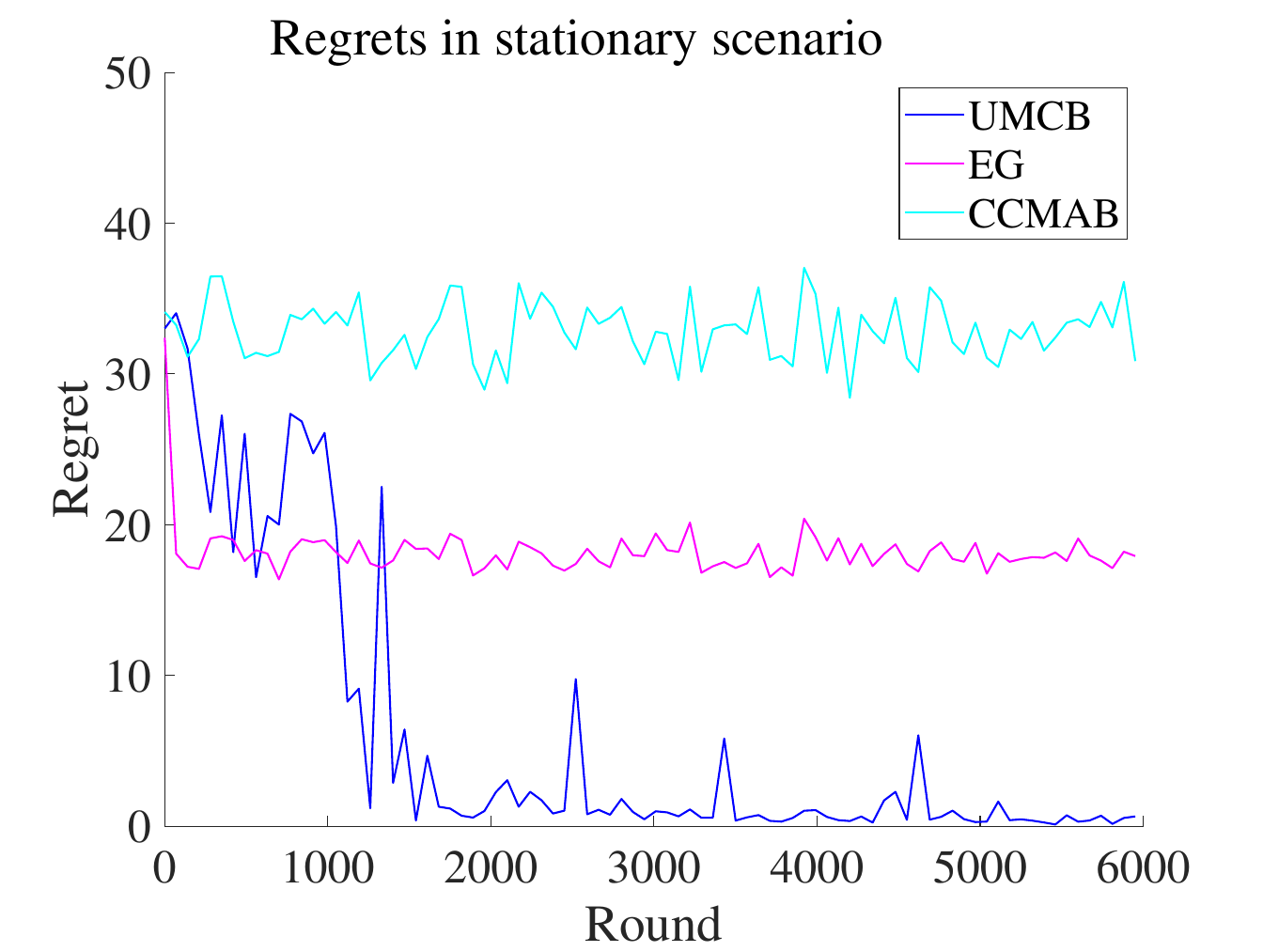}
        \label{fig:s_o1_reg}
    }
    \subfigure[With $U_2$]{
        \includegraphics[width=0.235\linewidth]{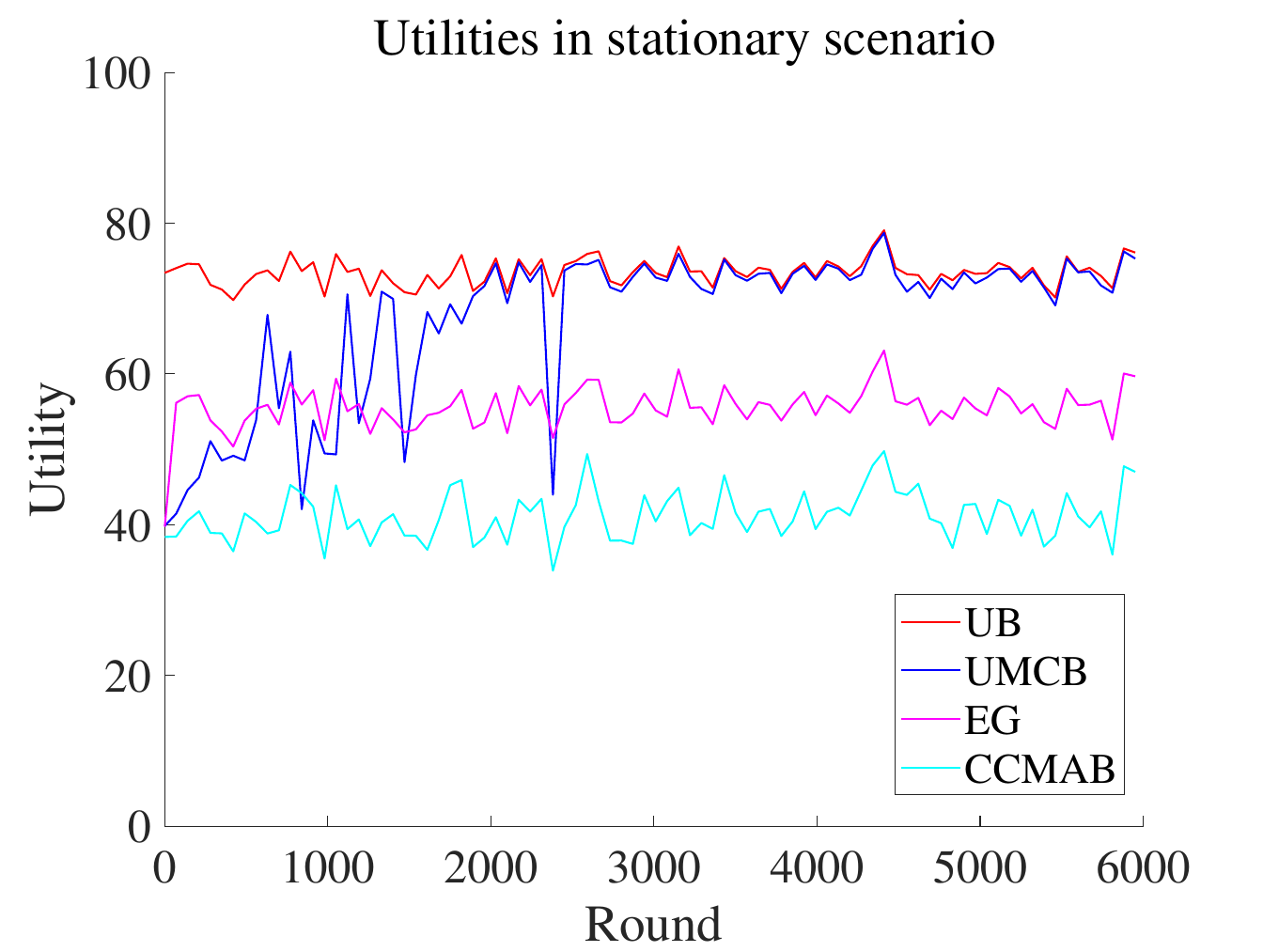}
        \label{fig:s_o3_res}
    
        \includegraphics[width=0.235\linewidth]{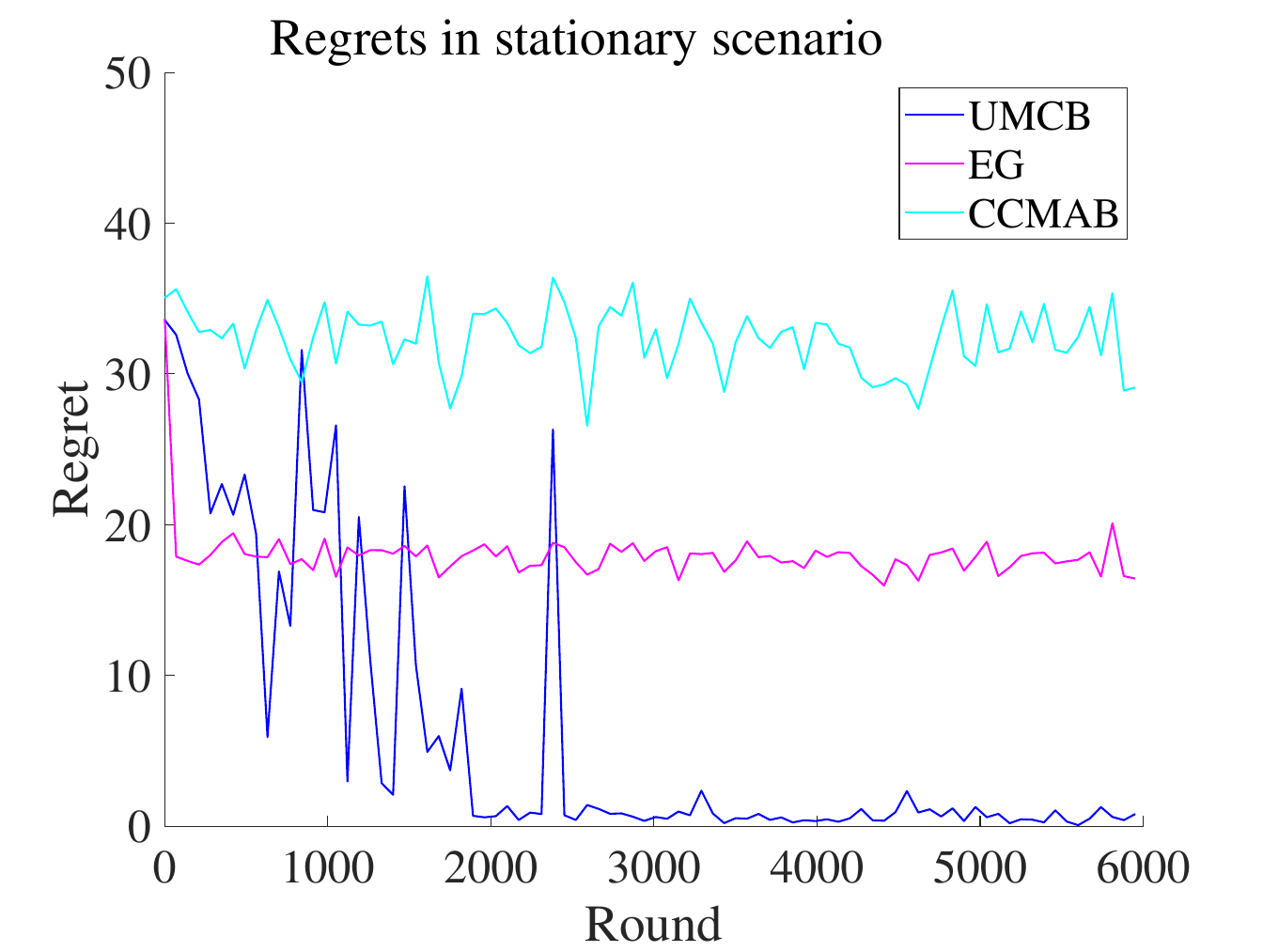}
        \label{fig:s_o3_reg}
    }
    \caption{Utilities and regrets with stationary channel 
    }
    \label{fig:sta_ur}
\end{figure*}

\subsubsection{Evaluation of Bandit Schemes}
Considering the unknown expected harvesting power, 
we use regret to indicate the gap between the evaluated algorithms and the upper bound.
In the experiments, {we set $Q=500$ and $T=6000$ in both stationary and non-stationary cases.} 
For the non-stationary case, we let the channel state distribution fluctuate in a random walking process by adding the dynamic changes to antenna gain of the sensors to simulate the changes of the antenna properties.
We denote the antenna gain of $s_i$ as $A_i$, with its initial value set as $A^0_i$, and add a randomly sampled drift $\Delta_A\in[-0.05A^0_i,0.05A^0_i]$ to $A_i$ in each round.
By default, the number of sensors $N$ is set to $20$.

We compare the proposed \ref{alg:UMCB} and UMCB-SW schemes with {\bf Upper Bound (UB)} and the following two strategies:
\begin{itemize}    
    \item {\bf Epsilon Greedy (EG)}:
    {it chooses charging policies randomly with probability $\epsilon_0$ and the best-known policy with probability $1-\epsilon_0$ in each slot. Here the exploring parameter $\epsilon_0=\frac{1}{3}$} .
    \item {\bf CCMAB~\cite{chen2018contextual}}: it maps the problem to the contextual combinatorial bandit problem. 
    It is another intuitive way to solve the bandit problem and has a parameter $h_T$ to partition the context space. In our experiments, 
    CCMAB is only compared when $N\le20$ due to its extremely large memory space required ($(h_T)^ N$). 
    
\end{itemize}

\subsubsection{UB: Theoretical Upper Bound with Continuous Time}
\label{sec:exp_ub}
We consider an ideal scenario to obtain the upper bound of UMC problem, where the time is continuous, meaning that for each policy $\pi_j$, its corresponding charging time $t_j$ is a real number and we can change the policy at any time. We also assume that the battery capacity $Q$ is unlimited in this case.
Recall that $p_i(\pi_j)$ represents the harvested power of $s_i$ under policy $\pi_j$. 
The harvested energy of sensor $s_i$ at the end of the charging round thus is $\sum_{j=1}^{|\Gamma|}(t_j p_i(\pi_j))$.
We re-write the original UMC problem to the following problem in the ideal scenario.

\vspace{0.1in}
\noindent \mathbi{Objective:}
\begin{equation}
\label{fr:convex}
\max \sum_{i=1}^{{N}}{U(q_i)}-\sum_{i=1}^{{N}}{U(C_i)}
\end{equation}
\noindent $\mathbi{Subject to:}$
\begin{equation}
\begin{array}{ll}
 C_i=x_i-\int_0^{T_c}  p^c_i(t)\mathrm{d}t,~\forall i\in[N] 
\\
    q_i=C_i+\sum_{j=1}^{|\Gamma|}t_jp_i(\pi_j), ~~q_i \geq 0 , ~~\forall i\in[N]
   
\\
\sum_{j=1}^{|\Gamma|}t_j\leq T_c, ~~ t_j \ge 0,~~\forall j\in [|\Gamma|]

\end{array}
\end{equation}
In {\bf P1}, we still need to guarantee $q_i\geq 0$ to avoid any sensor from energy depletion, but we do not need to guarantee that the residual energy of sensor $s_i$ is always bounded by $Q$.
The sum of the allocated time for each charging policy $\sum_{j=1}^{|\Gamma|}t_j$ should not exceed $T_c$.
{\bf P1} is a continuous convex problem. It can be directly solved by the CVX toolbox to obtain the optimal charging strategy.{
However, {\bf P1} is unrealistic as $t_j$ is a continuous variable and the battery capacity $Q$ is assumed to be unlimited.

\subsection{Evaluation Results}
\subsubsection{Evaluation of Oracles}
Fig.~\ref{fig:off_o1} and Fig.~\ref{fig:off_o3} show the average utility of the compared schemes under $U_1$ and $U_2$, respectively.
From the results, we can observe that the performance of {\ref{alg:greedy}} closely aligns with {UB}. 
When $Q\le 200$, {\ref{alg:greedy}} can reach the upper bound by nearly charging all sensors to their full capacity. 
As $Q$ increases, it becomes increasingly difficult to fully charge all the sensors in a round, resulting in a gap between {UB} and {\ref{alg:greedy}}.
The upper bound, while theoretically ideal, is practically unachievable as it assumes continuous time and unlimited battery. 
The performance of each scheme is also not pre-determined due to the inclusion of channel state variations in the simulation.
With $U_1$, {\ref{alg:greedy}} outperforms {GMQ} by at least $54.8\%$ when $Q<=300$ and at least $24\%$ when $Q$ is getting larger.
With $U_2$, \ref{alg:greedy} outperforms {GMQ} by at least $37.7\%$ when $Q<=300$ and at least $20.9\%$ when $Q$ is larger.

\subsubsection{Evaluation of Bandit Schemes}

\begin{figure*}[t]
    \centering
    \begin{minipage}[t]{0.49\linewidth}
        \centering
        \subfigure[with $U_1$]{
            \includegraphics[width=0.465\linewidth]{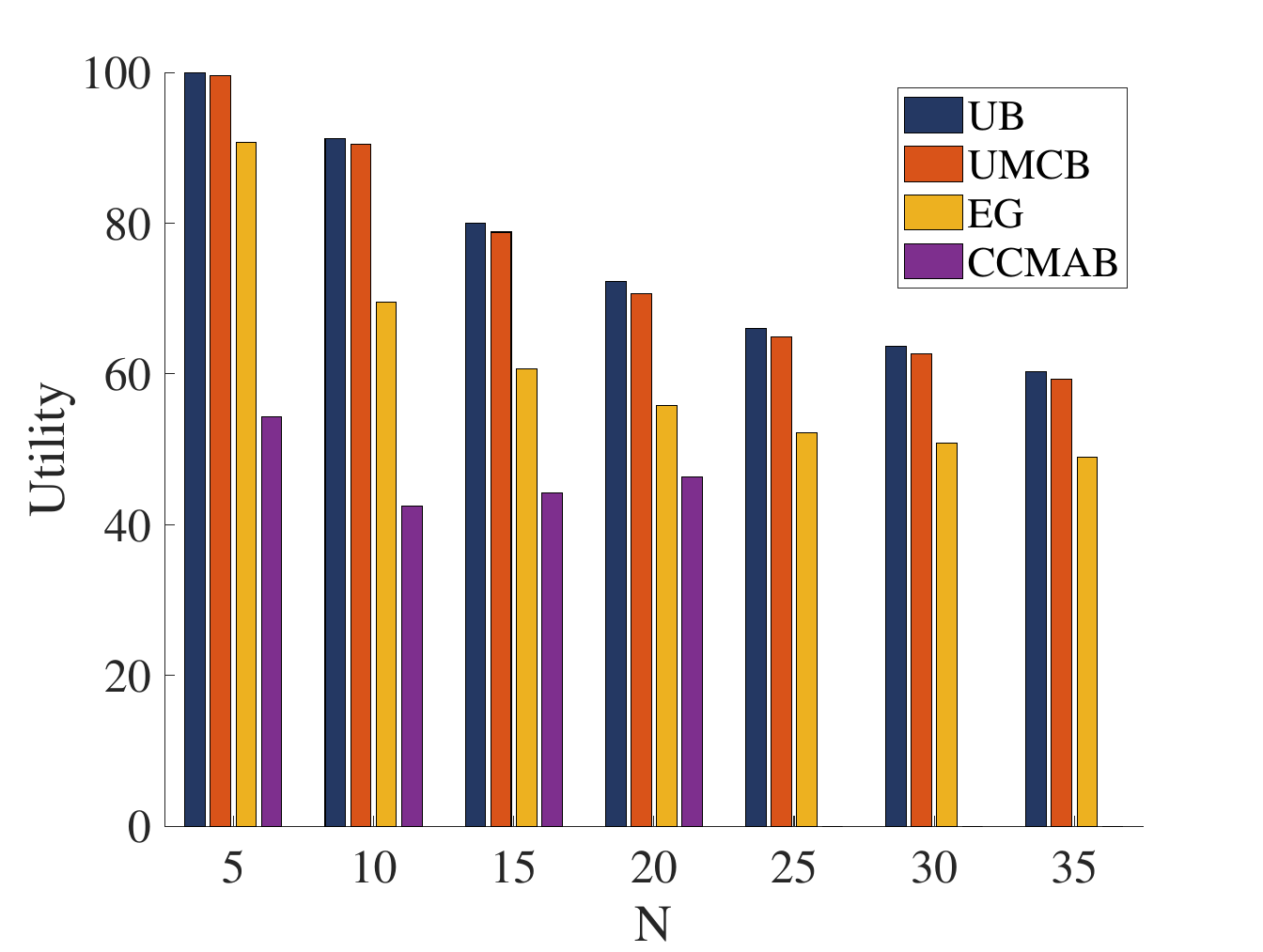}
            \label{fig:s_o1_n}
        }
        \subfigure[with $U_2$]{
            \includegraphics[width=0.465\linewidth]{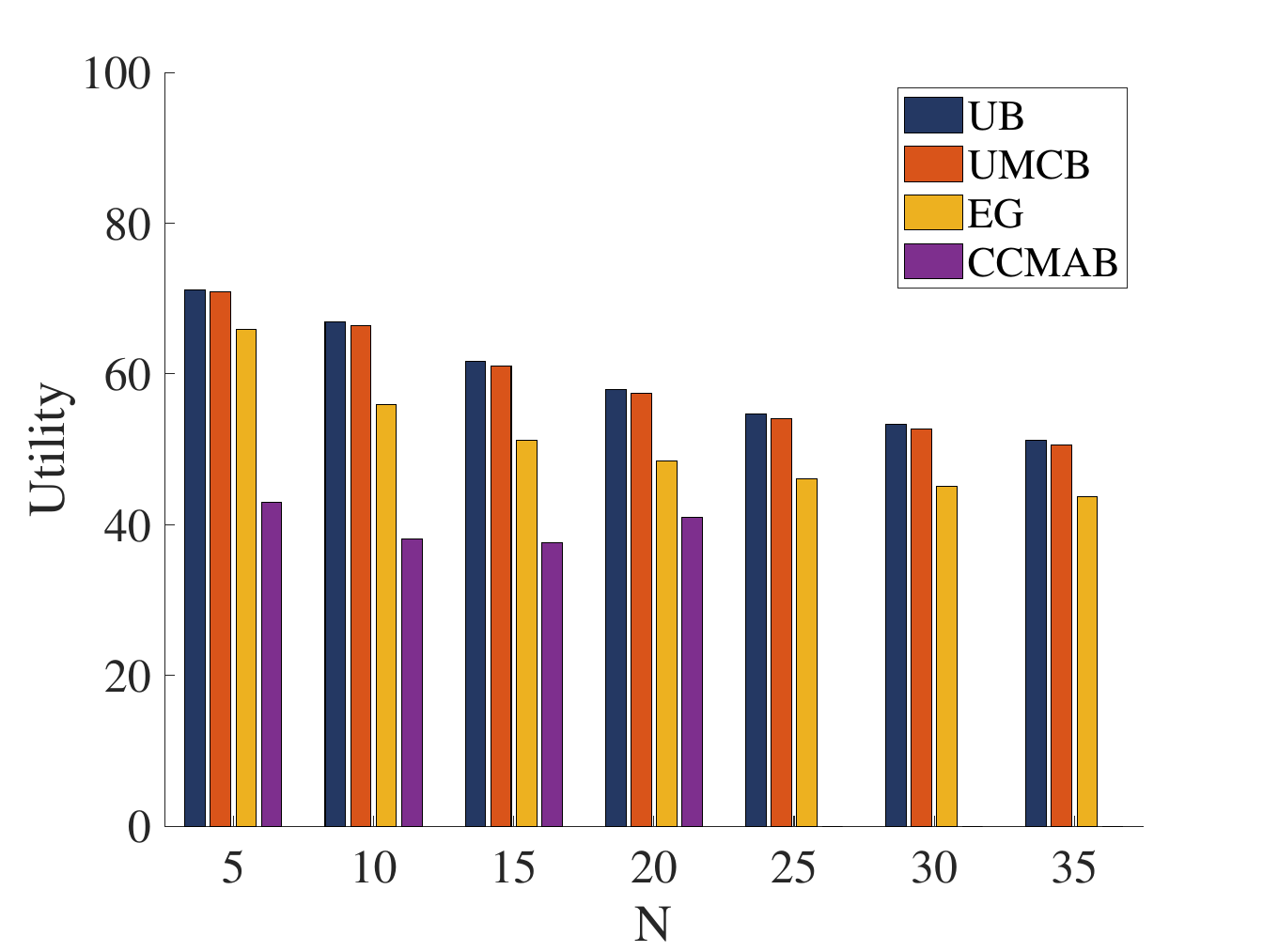}
            \label{fig:s_o3_n}
        }
        \caption{{Utility under stationary channel with varied $N$}}
        \label{fig:on_n}
    \end{minipage}
    \hfill
    \begin{minipage}[t]{0.49\linewidth}
        \centering
        \subfigure[with $U_1$]{
            \includegraphics[width=0.465\linewidth]{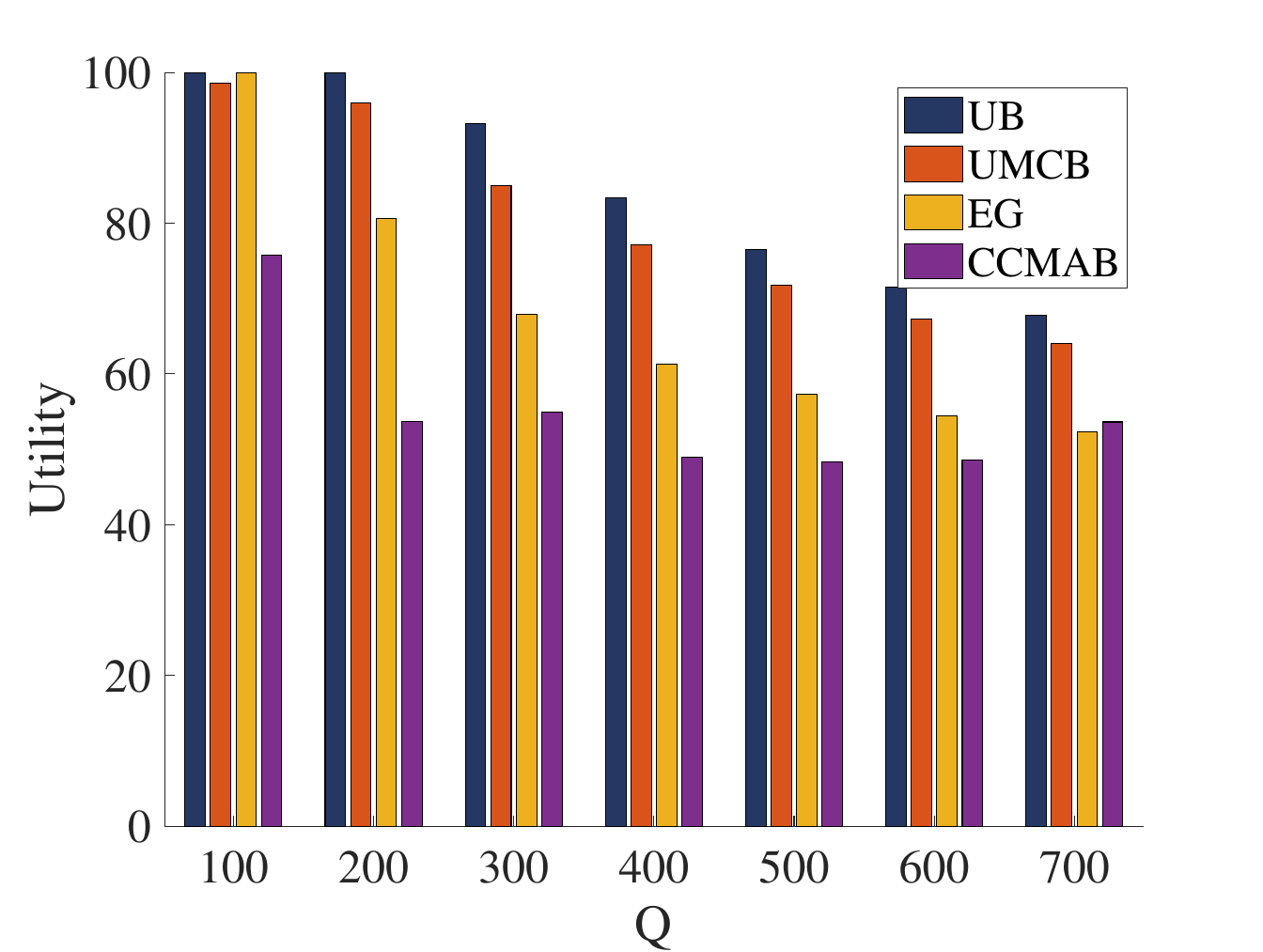}
            \label{fig:s_o1_q}
        }
        \subfigure[with $U_2$]{
            \includegraphics[width=0.465\linewidth]{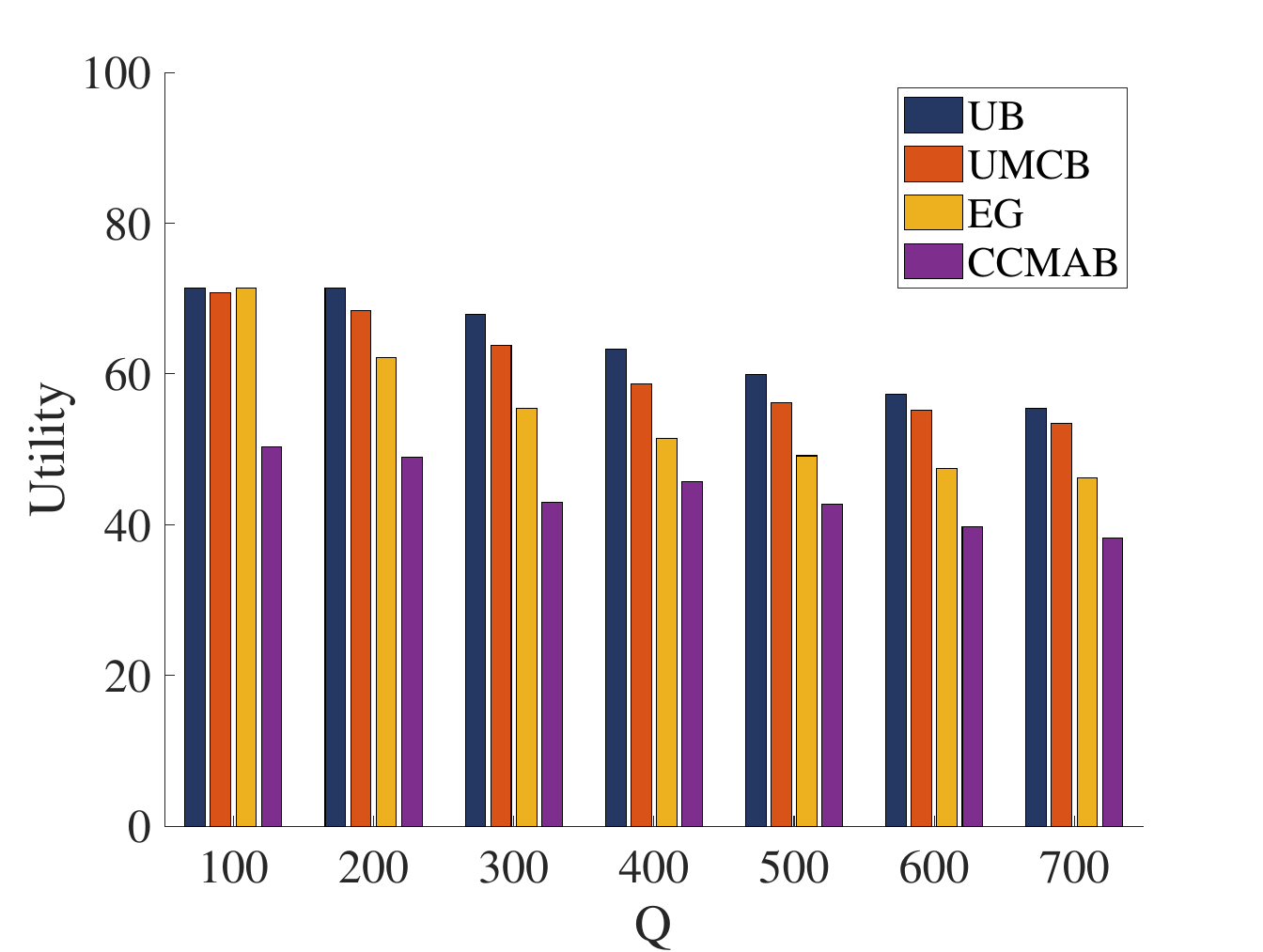}
            \label{fig:s_o3_q}
        }
        \caption{{Utility under stationary channel with varied $Q$}}
        \label{fig:on_q}
    \end{minipage}
\end{figure*}
\begin{figure*}[t]
    \centering
    \subfigure[With $U_1$]{
        \includegraphics[width=0.235\linewidth]{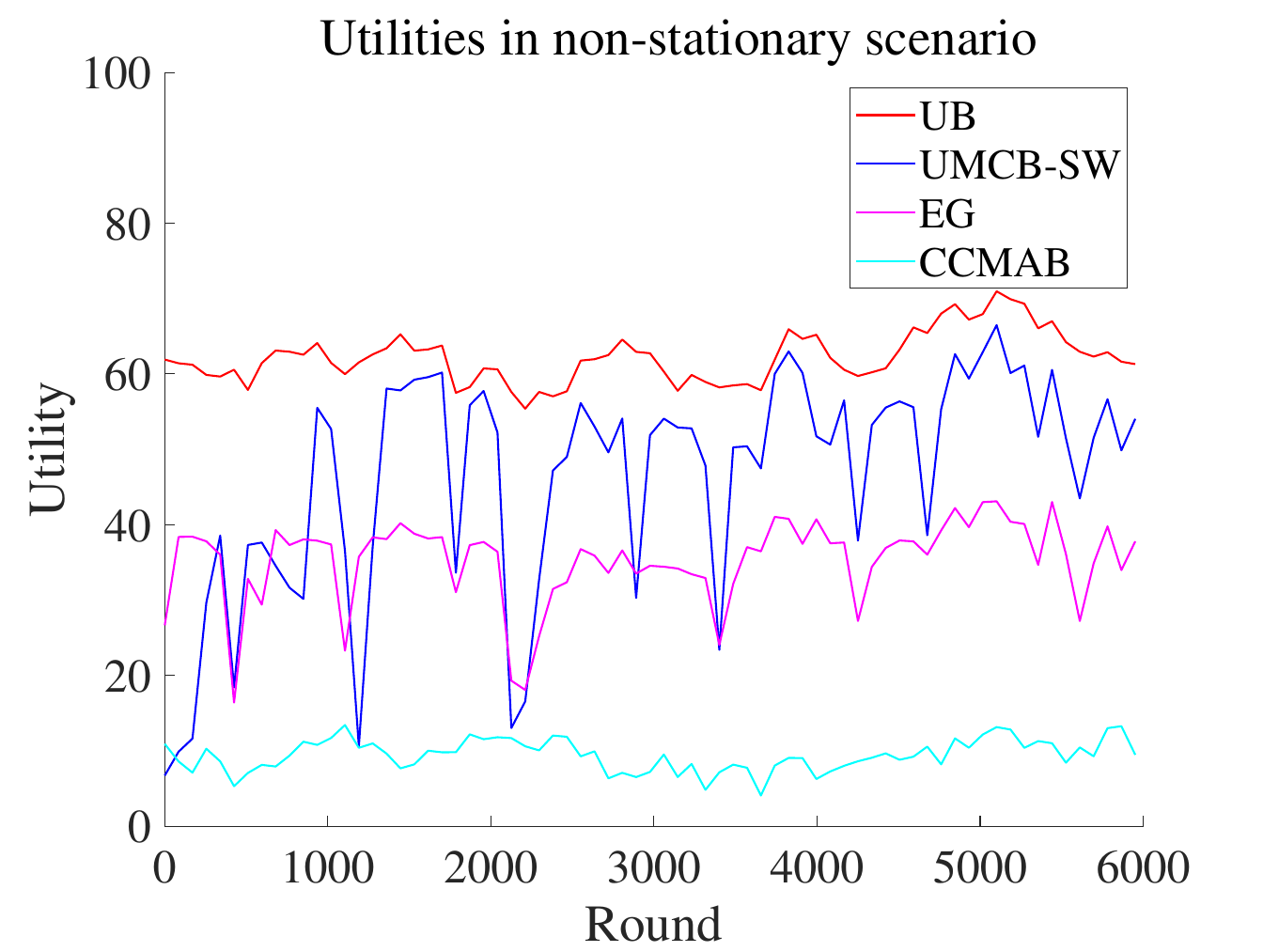}
        \label{fig:n_o1_res}
    
        \includegraphics[width=0.235\linewidth]{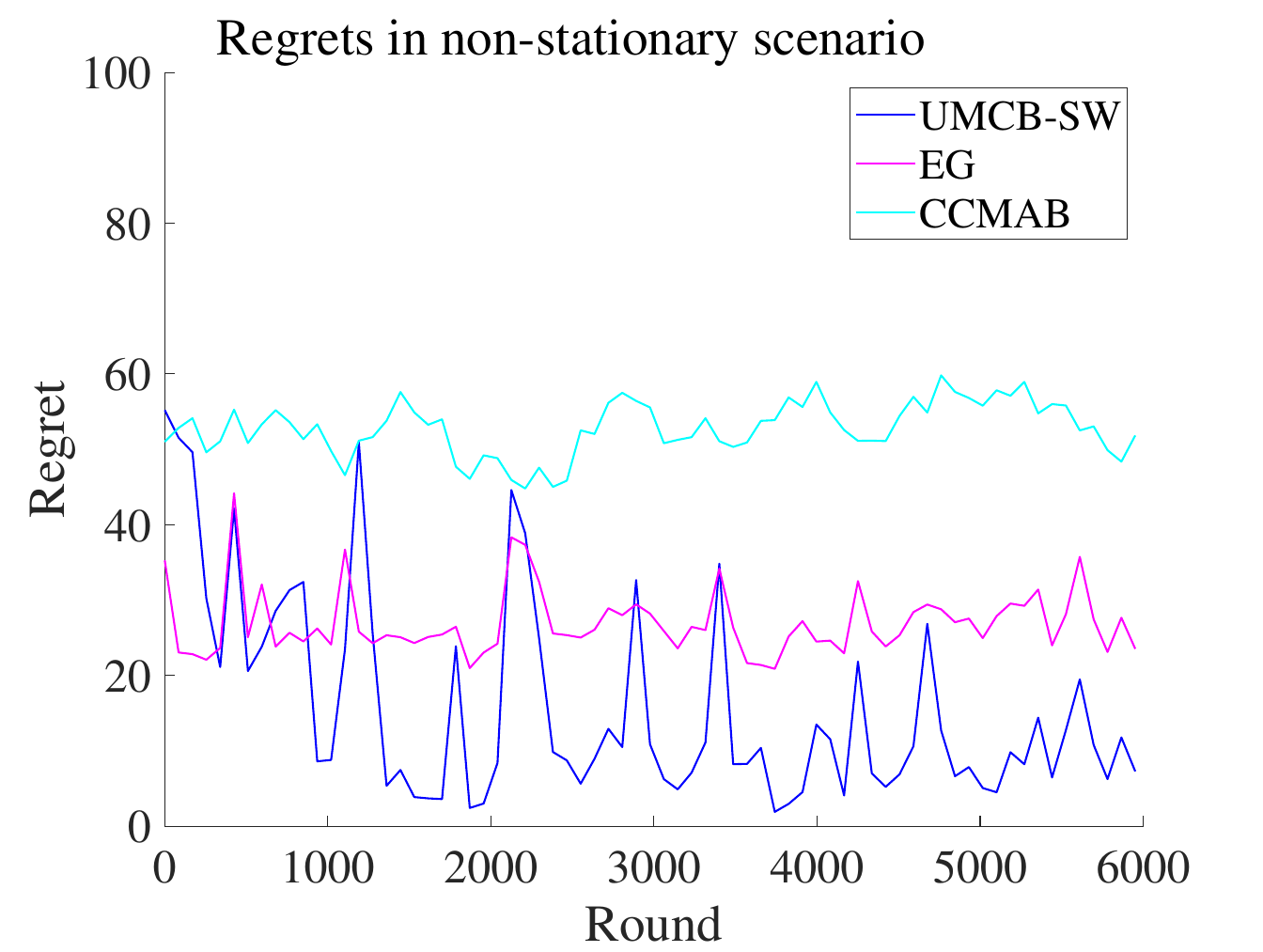}
        \label{fig:n_o1_reg}
    }
    \subfigure[With $U_2$]{
        \includegraphics[width=0.235\linewidth]{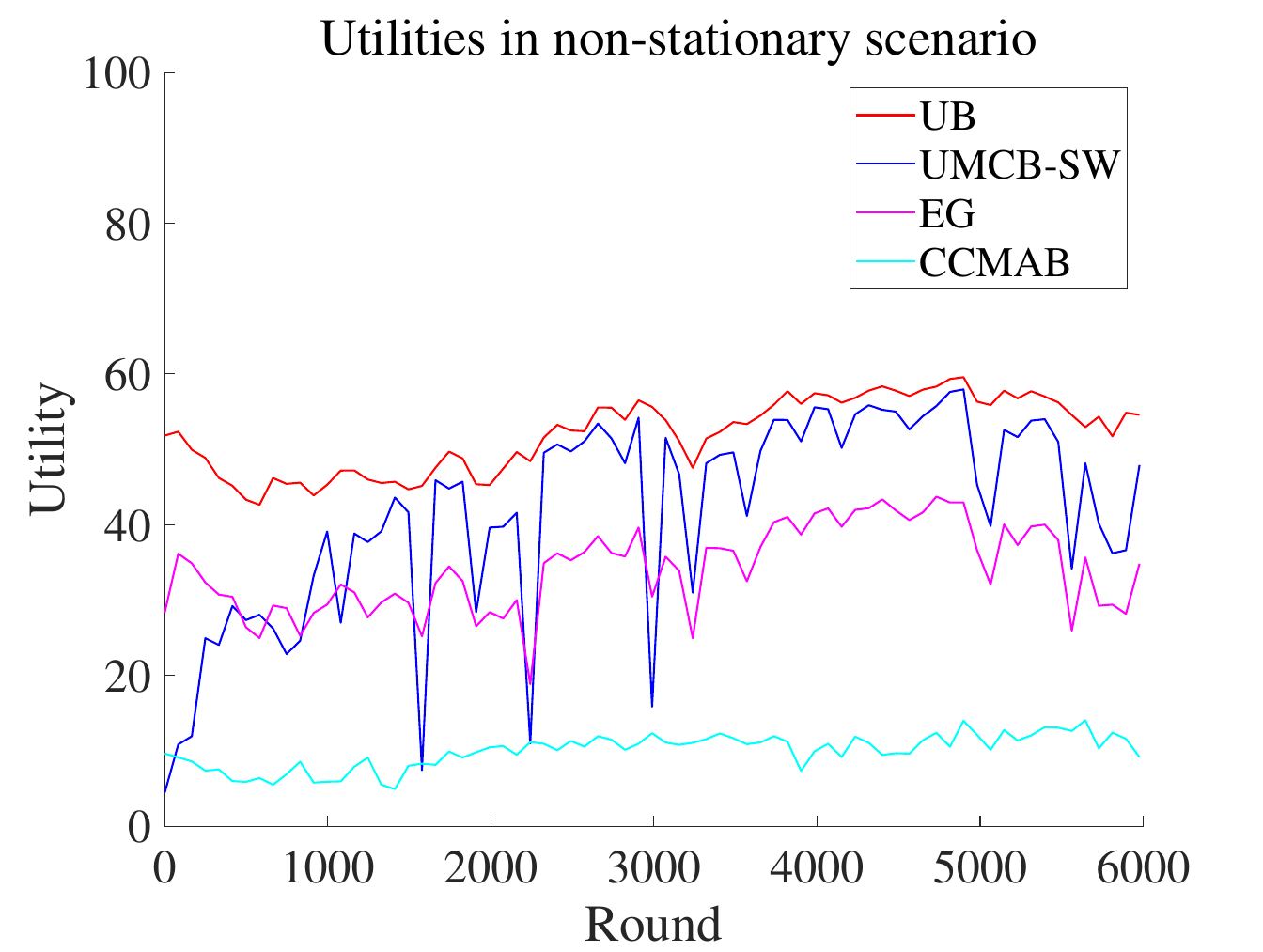}
        \label{fig:n_o3_res}
    
        \includegraphics[width=0.235\linewidth]{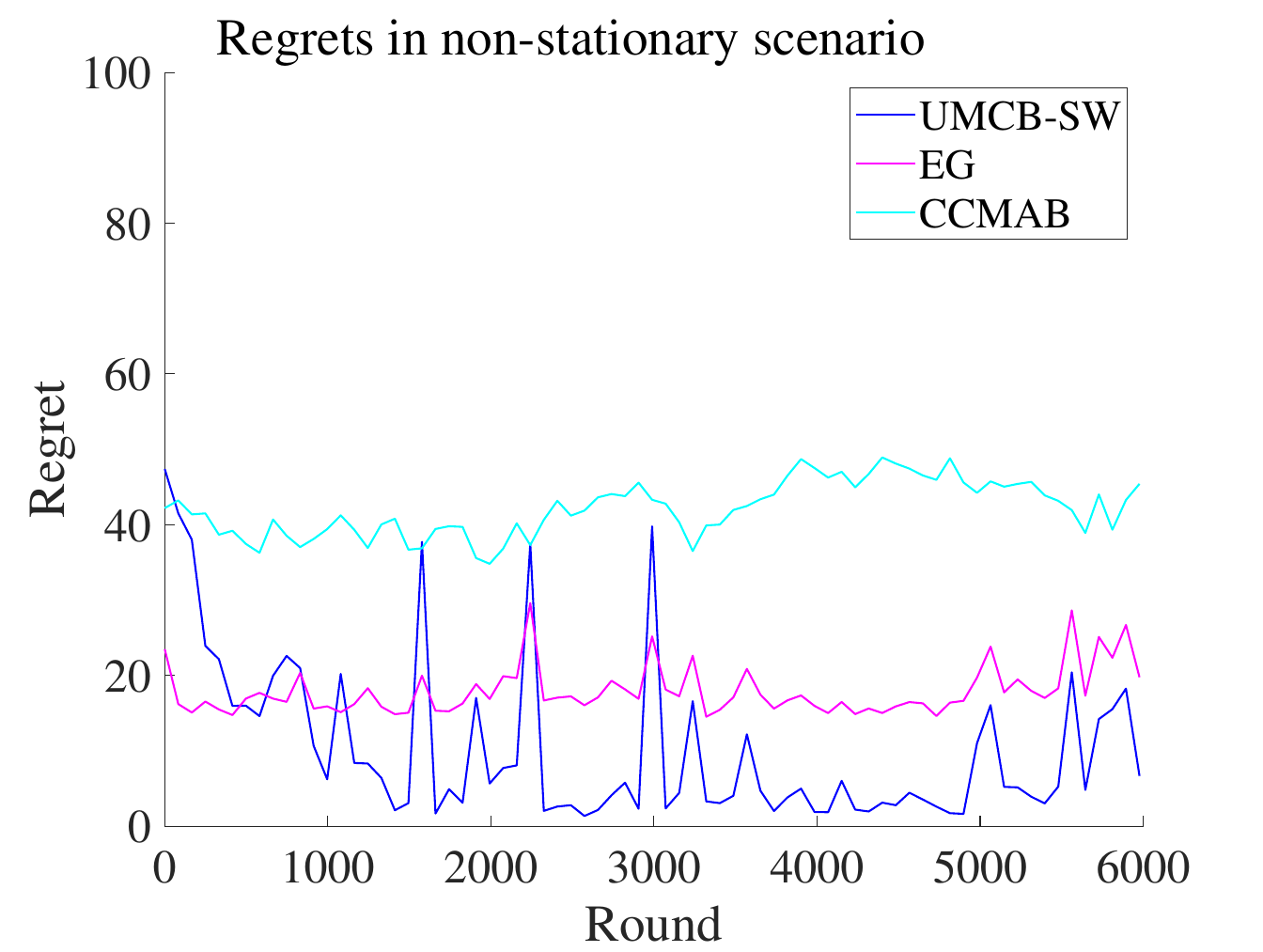}
        \label{fig:n_o3_reg}
    }
    \caption{Utilities and regrets with non-stationary channel}
    \label{fig:non_ur}
\end{figure*}

Fig.~\ref{fig:s_o1_res} and Fig.~\ref{fig:s_o3_reg} compares the utilities and regrets among the four schemes (UB, UMCB, EG and CCMAB) under stationary channels with utility functions $U_1$ and $U_2$, respectively. {Each round corresponds to a charging-decision epoch.}
In Fig.~\ref{fig:s_o1_reg}, the regret of UMCB decreases in fluctuation in the first $1200$ rounds, which indicates that it is the learning phase of the algorithm.
Afterwards, \ref{alg:UMCB} can effectively get close to the upper bound and outperforms other schemes.
When $t\approx2500$ and $t\approx3400$, the regret of \ref{alg:UMCB} has a sudden increase, which shows that it keeps exploring the environment where the confidence radius for other arms has been updated and increased when $T$ grows.
On the other hand, {CCMAB} can hardly balance between exploration and exploitation because its strategy space is too large to explore, though its regret bound is theoretically convergent.
Similar observation can be developed with utility function $U_2$ in Fig.~\ref{fig:s_o3_res}.

Fig.~\ref{fig:on_n} compares the average utility among the four schemes with varied $N$ under stationary channel.
From the results, we can observe that UMCB is always convergent after a certain number of rounds, and can achieve at least $98.1\%$ and $98.3\%$ of the upper bound with $U_1$ and $U_2$ respectively.
{The utility decreases with larger $N$ due to intensified competition for limited charging resources, which reduces each node’s charging opportunity.}
In short, UMCB outperforms EG and CCMAB in all scenarios.
Occasionally, when the regret value approaches nearly zero, the explanation could be that the majority of sensors have been fully charged at the end or the algorithm has learned well.

Fig.~\ref{fig:on_q} compares the average utility among the schemes with varied $Q$.
When $Q\le100$, both UMCB and EG can almost have all sensors fully charged in each round.
When $Q$ increases to be larger than $300$, UMCB outperforms EG and CCMAB by at least $23.3\%$ and $20.7\%$, respectively, under $U_1$ and by at least $15.2\%$ and $39.4\%$, respectively, under $U_2$.
In the meantime, UMCB can achieve at least $94.1\%$ and $90.1$ of the upper bound under $U_1$ and $U_2$, respectively.

   

Fig.~\ref{fig:non_ur} compares the utilities and regrets among four schemes (UB, UMCB-SW, EG and CCMAB) under non-stationary environments with utility function $U_1$ and $U_2$, respectively. From the results, we can observe that the upper bound remains relatively stable in non-stationary environment. 
From Fig.~\ref{fig:n_o1_res}, we can see that UMCB-SW gradually improves its utility in the initial $1400$ rounds and outperforms other schemes after $1300$ rounds.
UMCB-SW approaches the upper bound more slowly compared to its performance in stationary environments but can track it consistently although the channel states keep changing. On the other hand, when the antenna properties changes, the empirical solution will possibly not be as effective as before, which causes the increased regret.
This can be observed when $t\approx2200$, $t\approx3400$ and $t\approx4300$ as shown in Fig.~\ref{fig:n_o1_reg} and $t\approx1600$, $t\approx2200$ and $t\approx3000$ as shown in Fig.~\ref{fig:n_o3_reg} as well.
After that, the regret of UMCB-SW drops back to a consistently low level, indicating that UMCB-SW can adapt to the non-stationary environment quickly.

\section{Conclusion}\label{sec:conclusion}

\qiu{In this paper, we proposed novel temporal-spatial bandit-based charging algorithms to charge WP-IoT systems with mobile charger(s) under real-time constraints using beamforming.
The proposed charging policy GUA with known CSI is proven to achieve a constant approximation ratio.
To address real-time uncertainties and dynamics, we develop two online bandit algorithms, UMCB and UMCB-SW, that operate without prior knowledge and have bounded regret.
Our framework explicitly considers strict per-round time constraints from limited charger battery, ensuring real-time-feasible scheduling and charging.
Extensive simulations show that the proposed algorithms approach the upper bound and outperform state-of-the-art methods in efficiency.}



\clearpage
\bibliographystyle{IEEEtran}
\bibliography{sample.bib}

\end{document}